\documentclass[runningheads,orivec,11pt]{llncs}
\usepackage[margin=1in]{geometry}

 \pdfoutput=1

\usepackage{caption}

\usepackage[utf8]{inputenc}
\usepackage{cite}
\usepackage{subfig}
\usepackage{xcolor}

\usepackage{xspace}
\usepackage{enumerate}

\usepackage{microtype}

\usepackage{amsmath}
\usepackage{amssymb}
\usepackage{amstext}
\usepackage{amsopn}
\usepackage{hyperref}
\usepackage{graphicx}
\usepackage[noline,boxruled,noend,algonl,boxed]{algorithm2e}
\usepackage{paralist}
\usepackage[left, pagewise]{lineno}
\usepackage[textsize=footwww.google.notesize]{todonotes}
\SetVlineSkip{0pt}  
\newtheorem{observation}{Observation}
\def\comment#1{}%
\def\withcomments{%
  \newcounter{mycommentcounter}%
   \def\comment##1{\refstepcounter{mycommentcounter}%
    \ifhmode%
     \unskip%
     {\dimen1=\baselineskip \divide\dimen1 by 2 %
       \raise\dimen1\llap{\tiny
	{-\themycommentcounter-}}}\fi%
     \marginpar[{\renewcommand{\baselinestretch}{0.8}%
       \hspace*{-2em}\begin{minipage}{1.5\marginparwidth}\footnotesize%
[\themycommentcounter]:%
\raggedright ##1\end{minipage}}]{\renewcommand{\baselinestretch}{0.8}%
       \begin{minipage}{1.5\marginparwidth}\footnotesize%
[\themycommentcounter]: \raggedright%
##1\end{minipage}}}%
  }

\newcommand{\dg}{^\circ}
\newcommand{\mn}{^-}
\newcommand{\eps}{\varepsilon}

\newcommand{\ray}[2]{\ensuremath{\textnormal{ray}(#1,\vec{#2})}}%
\newcommand{\hp}[2]{\ensuremath{h_{#1}^{#2}}}
\DeclareMathOperator{\axis}{axis}

\DeclareMathOperator{\proc}{p}
\DeclareMathOperator{\apex}{apex}
\DeclareMathOperator{\procn}{N_p}
\DeclareMathOperator{\neighb}{N}
\DeclareMathOperator{\polytope}{polytope}
\newcommand{\true}{\textnormal{true}}
\newcommand{\false}{\textnormal{false}}
\newcommand{\optangle}{\angle_\textnormal{optimal}}
\newcommand{\gP}[1]{\ensuremath{\mathcal P^{#1}}}
\newcommand{\distT}[2]{\ensuremath{d_T(#1,#2)}}
\newcommand{\turnccw}[2]{\ensuremath{\angle_\textnormal{ccw}(\vec{#1},\vec{#2})}}
\newcommand{\optprobref}{\eqref{eq:opt-problem}\xspace} %
\withcomments

\let\doendproof\endproof
\renewcommand\endproof{~\hfill\qed\doendproof}
\title{Euclidean Greedy Drawings of Trees
}
 
\author{Martin N\"ollenburg and Roman Prutkin}  
\authorrunning{N\"ollenburg \and Prutkin}
\institute{Institute of Theoretical Informatics, Karlsruhe Institute of Technology, Germany}
\begin{document}

\maketitle

\begin{abstract}
Greedy embedding (or drawing) is a simple and efficient strategy to route messages in wireless sensor networks.
For each source-destination pair of nodes~$s,t$ in a greedy embedding there is always a neighbor~$u$ of~$s$ that is closer to~$t$ according to some distance metric. 
The existence of greedy embeddings in the Euclidean plane $\mathbb{R}^2$ is known for certain graph classes such as 3-connected planar graphs.
We completely characterize the trees that admit a greedy embedding in~$\mathbb R^2$.
This answers a question by Angelini et al.~(Graph Drawing 2009)
and is a further step in characterizing the graphs that admit Euclidean greedy  embeddings.

\end{abstract} 

\section{Introduction}
Message routing in wireless ad-hoc and sensor networks cannot apply  the same established global hierarchical routing schemes that are used, e.g., in the Internet Protocol. A family of alternative routing strategies in wireless networks known as \emph{geographic routing} uses node locations as addresses instead. The  \emph{greedy routing} protocol simply passes a message at each node to a neighbor that is closer to the destination. Two problems with this approach are (i) that sensor nodes typically are not equipped with GPS receivers due to their cost and energy consumption and (ii) that even if nodes know their positions messages can get stuck at voids, where no node closer to the destination exists.

An elegant strategy to tackle these issues was proposed by Rao et al~\cite{Rao2003}. 
It maps nodes to virtual rather than geographic coordinates, on which the standard greedy routing is then performed. 
A mapping that always guarantees successful delivery is called a \emph{greedy embedding} or \emph{greedy drawing}.

The question about the existence of greedy embeddings for various metric spaces and classes of graphs 
has attracted a lot of interest. %
Papadimitriou and Ratajczak~\cite{Papadimitriou2005} conjectured that every 3-connected planar graph 
admits a greedy embedding into the Euclidean plane. 
Dhandapani~\cite{Dhandapani2010} proved that every 3-connected planar triangulation has a greedy drawing. 
The conjecture by Papadimitriou and Ratajczak itself has been proved independently 
by Leighton and Moitra~\cite{lm-srgems-10} and Angelini et al.~\cite{angelini2010algorithm}. 
Kleinberg~\cite{Kleinberg2007} showed that every connected graph has a greedy embedding in the hyperbolic plane.
\par
Since efficient use of storage and bandwidth are crucial in wireless sensor networks, 
virtual coordinates should require only few, i.e., $O(\log n)$, bits in order to keep message headers small.
Greedy drawings with this property are called \emph{succinct}. 
Eppstein and Goodrich proved the existence of succinct greedy drawings for 3-connected planar graphs 
in~$\mathbb R^2$~\cite{Eppstein2011}, 
and Goodrich and Strash~\cite{Goodrich2009} showed it for any connected graph in the hyperbolic plane. 
Wang and He~\cite{Wang2012} used a custom distance metric and constructed convex, 
planar and succinct drawings for 3-connected planar graphs using Schnyder realizers~\cite{Schnyder1990}.
\par
It has been known that not all graphs admit a Euclidean greedy drawing in the  plane, e.g.,~$K_{k,5k+1}$ ($k \geq 1$) has no such drawing~\cite{Papadimitriou2005}, including the tree~$K_{1,6}$.
The (non-)existence of a greedy drawing for some particular tree is used in a number of proofs as an intermediate result. %
Leighton and Moitra~\cite{lm-srgems-10} showed that a graph containing at least six \emph{pairwise independent irreducible triples} 
(e.g., the complete binary tree containing~31 nodes) cannot have a greedy embedding.
They used this fact to present a planar graph that admits a greedy embedding,
although none of its spanning trees does.
We give examples of graphs with no greedy drawing that contain at most \emph{five} such triples.
Further, there are greedy-drawable trees
that have no succinct Euclidean greedy drawing~\cite{abf-sgdae-10}.
\par 
\emph{Self-approaching drawings}~\cite{acglp-sag-12} are a subclass of greedy drawings with the additional constraint that for any pair of nodes there is a path~$\rho$ that is distance decreasing not just for the node sequence of~$\rho$ but for any triple of intermediate points on the edges of~$\rho$.
Alamdari et al.~\cite{acglp-sag-12} gave a complete characterization of trees admitting self-approaching drawings.
Since self-approaching drawings are greedy, all trees with a self-approaching drawing are greedy-drawable.
However, there exist numerous trees that admit a greedy drawing, but no self-approaching one, and the characterization of those trees turns out to be quite complex. \par
\subsubsection{Our contributions.} 
We give the first complete characterization of all trees that admit a greedy embedding in~$\mathbb R^2$ 
with the Euclidean distance metric.
This solves the corresponding open problem stated by Angelini et al.~\cite{abf-sgdae-10}
and is a further step in characterizing the graphs that have greedy embeddings.
For any given tree~$T$ and an edge~$e$ of~$T$ separating~$T$ into~$T_1$ and~$T_2$,
we calculate a tight upper bound on the opening angle of a cone 
formed by perpendicular bisectors of edges of~$T_1$,
in which~$T_2$ is contained 
in any greedy embedding
in time linear in the size of~$T_1$.
We then show that deciding whether~$T$ has a greedy embedding is equivalent to 
deciding whether there exists a valid angle assignment in a certain wheel polygon. This includes a non-linear constraint known as the \emph{wheel condition}~\cite{DiBattistaVismara1993}.
For most cases (all trees with maximum degree~4 and most trees with maximum degree~5) we are able to give an explicit solution to this problem,
which provides a linear-time recognition algorithm.
For trees with maximum degree~3 we give an alternative characterization by forbidden subtrees. %
For some trees with one degree-5 node we resort to using non-linear solvers.
For trees with nodes of degree~$\ge 6$ no greedy drawings exist. 
\par
Our proofs are constructive, however, we ignore the 
possibly exponential area requirements for our constructions.
This is justified by the aforementioned result that some trees require exponential-size greedy drawings~\cite{abf-sgdae-10}.
\par
\section{Preliminaries}
\label{sec:prelim}
In this section, we introduce the concept of the opening angle of a rooted subtree
and present relations between opening angles %
that will be crucial for the characterization
of greedy-drawable trees. We start with a number of lemmas on basic properties of opening angles and greedy drawings and sketch the main ideas of our characterization. This is followed by proving the \emph{shrinking lemma}, which serves as a main tool for our later constructions.
\par
Let $T=(V,E)$ be a tree. A \emph{straight-line drawing} $\Gamma$ of $T$
maps every node $v\in V$ to a point in the plane $\mathbb R^2$ and every edge $uv \in E$ to the line segment between its endpoints. We say that $\Gamma$ is \emph{greedy} if for every pair of nodes $s,t$ there is a neighbor $u$ of $s$ with $|ut| < |st|$, where $|st|$ is the Euclidean distance between points~$s$ and~$t$. To ease notation we identify nodes with points and edges with line segments. Furthermore we assume that all drawings are straight-line drawings.

It is known that for a greedy drawing $\Gamma$ of $T$ any subtree of $T$ is  represented in $\Gamma$ by a greedy subdrawing~\cite{abf-sgdae-10}. We define the \emph{axis} of an edge $uv$ as its perpendicular bisector. Let~$h_{uv}^u$ denote the open half-plane bounded by the axis of~$uv$ and containing~$u$.
Let~$T_{uv}^u$ be the subtree of~$T$ containing~$u$ obtained from~$T$ by removing~$uv$. Angelini et al.~\cite{abf-sgdae-10} showed that in a greedy drawing of $T$ every subtree~$T_{uv}^u$ is contained in~$h_{uv}^u$. The converse is also true.
\begin{lemma}
 Let~$\Gamma$ be a drawing of~$T$ with $T_{uv}^u \subseteq h_{uv}^u$ $\forall uv \in E$.  Then,~$\Gamma$ is greedy. \label{lem:bisectors-dont-cross}
\end{lemma}
\begin{proof}
 For $s,t \in V$ let~$u$ be the neighbor of~$s$ on the unique $s$-$t$ path in~$T$.  Since $t \in T_{su}^u \subseteq h_{su}^u$, we have $|ut|<|st|$.
\end{proof}
Angelini et al.~\cite{abf-sgdae-10} further showed that greedy tree drawings are always planar and that in any greedy drawing of $T$ the angle between two adjacent edges must be strictly greater than $60^\circ$. Thus $T$ cannot have a node of degree $\ge 6$.

Let $\ray{u}{uv}$ denote the ray with origin~$u$ and direction~$\vec{uv}$.
For $u,v \in V$, let $\distT{u}{v}$ be the length of the $u$-$v$ path in~$T$.
For vectors~$\vec{ab}$, $\vec{cd}$, let~$\turnccw{ab}{cd}$ denote the counterclockwise turn (or turning angle) from~$\vec{ab}$ to~$\vec{cd}$.
\begin{lemma}[Lemma 7 in~\cite{abf-sgdae-10}]\label{lem:lem-slope}
 Consider two edges~$uv$ and~$wz$ in a greedy drawing of $T$, such that the path from~$u$ to~$w$
 does not contain~$v$ and~$z$.
 Then, the rays~$\ray{u}{uv}$ and~$\ray{w}{wz}$ diverge; see Fig.~\ref{fig:prelim:lem-slope}.
\end{lemma}
\newcommand{\lemDegTwoContractText}{}
\begin{lemma}
 Let~$\Gamma$ be a greedy drawing of~$T$, $v \in V$, $\deg(v) = 2$, $\neighb(v) = \{ u,w \}$ the only two neighbors of~$v$, and~$T' = T - \{ uv, vw \} + \{ uw \}$.  The drawing~$\Gamma'$ induced by replacing segments~$uv$, $vw$ by~$uw$ in~$\Gamma$ is also greedy.
 \label{lem:deg2-contract}
\end{lemma}
\begin{proof}
 For~$x,y$ in~$T'$, let~$\rho'$ and $\rho$ be the $x$-$y$-paths in~$T'$ and in~$T$, respectively.  If~$\rho \neq \rho'$, then~$v \in \rho$.
 Since distance to~$y$ decreases along~$\rho$, it also decreases along~$\rho'$.
 Hence,~$\Gamma'$ is greedy.
\end{proof}
Next we generalize some concepts from Leighton and Moitra~\cite{lm-srgems-10}. For $k=3,4,5$, we define an \emph{irreducible $k$-tuple} as a $k$-tuple of nodes~$(b_1, \dots, b_k)$ 
 in a graph~$G=(V,E)$, such that~$\deg(b_1)=k$, 
 $b_1 b_2, b_1 b_3, \dots, b_1 b_k \in E$ (we call these $k-1$ edges \emph{branches} of the $k$-tuple) and 
 removing any branch~$b_1 b_j$ disconnects the graph.
 A $k$-tuple $(b_1, \ldots, b_k)$ and an $l$-tuple $(x_1, \ldots, x_l)$ are \emph{independent}, if
 $\{ b_1, \ldots, b_k \} \cap \{ x_1, \ldots, x_l \} = \emptyset$, and
  deleting all the branches keeps~$b_1$ and~$x_1$ connected.
\par
Let~$\Gamma$ be a greedy drawing of~$T$. We shall consider subtrees~$T_i=(V_i, E_i)$ of $T$, such that~$T_i$ has root~$r_i$,
$\deg(r_i)=1$ in~$T_i$ and~$v_i$ is the neighbor of~$r_i$ in~$T_i$. We define the \emph{polytope of a rooted subtree} $T_i$ as $\polytope(T_i) = \bigcap \{ h_{uw}^w \mid uw \in E_i, uw \neq r_i v_i, \; \distT{w}{r_i} < \distT{u}{r_i}\}$.

\begin{definition}[Extremal edges]
 For~$j=1,2$, let~$a_j b_j \neq v_i r_i$ be an edge of~$T_i$, $\distT{a_j}{r_i} < \distT{b_j}{r_i}$, such that direction~$\vec{a_j b_j}$
 is closest to~$\vec{v_i r_i}$ clockwise for~$j=1$ and counterclockwise for~$j=2$.
 We call edges~$a_j b_j$ \emph{extremal}.
\end{definition}
Note that by Lemma~\ref{lem:lem-slope}, $\ray{a_j}{a_j b_j}$ and~$\ray{v_i}{v_i r_i}$ diverge.
In the following two definitions, let~$e_j = a_j b_j$, $j=1,2$ be the extremal edges of~$T_i$.
\begin{definition}[Open angle]
Let~$\turnccw{a_1 b_1}{a_2 b_2} > 180\dg$. Then, $\polytope(T_i)$ is unbounded,
and we say that $T_i$ is drawn with an \emph{open angle}.
 \begin{compactitem}
  \item[(a)] If~$a_1 b_1 \not\subseteq \hp{a_2 b_2}{b_2}$ and~$a_2 b_2 \not\subseteq \hp{a_1 b_1}{b_1}$,
   we define 
   $ \angle T_i = \hp{a_1 b_1}{a_1} \cap \hp{a_2 b_2}{a_2}.$
   Let~$x_i$ be the intersection of~$\axis(e_1)$ and~$\axis(e_2)$.
   We set~$\apex(\angle T_i) = x_i$; see Fig.~\ref{fig:prelim:op_angle}.
  \item[(b)] If~$a_j b_j \subseteq \hp{a_k b_k}{b_k}$ for $j=1, k=2$ or $j=2, k=1$,
  let~$\angle T_i$ be the cone defined by moving the boundaries 
  of~$\hp{a_1 b_1}{a_1}$, $\hp{a_2 b_2}{a_2}$ to~$b_j$ ($b_k \in \angle T_i$),
  and~$\apex(\angle T_i) = b_j$.
 \end{compactitem}
 We call~$\angle T_i$ \emph{the opening angle} of~$T_i$ in~$\Gamma$ (orange in Fig.~\ref{fig:prelim:op_angle}).
 We write~$|\angle T_i| = \alpha$, where~$\alpha$ is the angle between the two rays of~$\angle T_i$.   
 \end{definition}
Obviously,~$\polytope(T_i) \subseteq \angle T_i$ in~(a).
This is also the case in~(b) by Observation~\ref{lem:move-halfplane}.
\begin{observation}
Let~$h$ be an open half-plane and~$p \notin h$. Let~$h'$ be the half-plane created by a parallel translation of the boundary of~$h'$ to~$p$. Then,~$h \subseteq h'$.
 \label{lem:move-halfplane}
\end{observation} 
\begin{figure}
 \hfill
  \subfloat[]{\includegraphics[scale = 1.0, page=2]{fig/full/prelim.pdf}
  \label{fig:prelim:lem-slope}
 }
 \hfill
 \subfloat[]{\includegraphics[scale = 1.0, page=1]{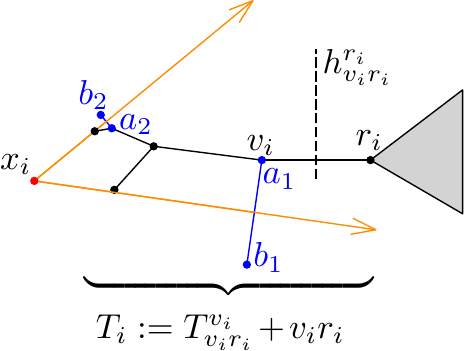}
  \label{fig:prelim:op_angle}
 } 
 \hfill\null
 \caption{
 \protect\subref{fig:prelim:lem-slope} Sketch of Lemma~\ref{lem:lem-slope}.
 \protect\subref{fig:prelim:op_angle} Subtree~$T_i$ with opening 
 angle~$\angle T_i$ (orange), 
 extremal edges~$a_1 b_1$, $a_2 b_2$ (blue) and apex $x_i$ (red).
 The subtree~$T_{v_i r_i}^{r_i}$ (gray triangle) must be contained in the half-plane~$\hp{v_i r_i}{r_i}$
 and the cone~$\angle T_i$. }
\end{figure}
\begin{definition}[Closed and zero angle]
 Let~$\turnccw{a_1 b_1}{a_2 b_2} < 180\dg$ (or $= 180\dg$).
 Let~$C_i = h_{a_1 b_1}^{a_1} \cap h_{a_2 b_2}^{a_2}$, and let~$p_j$ be the midpoint of~$e_j$.  We denote the part of~$C_i$ bounded by segment~$p_1 p_2$ containing~$r$ by~$\angle T_i$ and say that~$T_i$ is drawn with a \emph{closed (or zero) angle}; see Fig.~\ref{fig:prelim:closed}.
  We write~$|\angle T_i| < 0$ (or $=0$).
\end{definition}

We say that two subtrees~$T_1$, $T_2$ are \emph{independent}, if
$T_2\setminus \{ r_2 \} \subseteq T_{v_1 r_1}^{r_1}$ and~$T_1\setminus \{ r_1 \} \subseteq T_{v_2 r_2}^{r_2}$.
If~$T_1$ and~$T_2$ are independent, then 
$T_2 \setminus \{ r_2 \} \subseteq h_{v_1 r_1}^{r_1}$ and~$T_1 \setminus \{ r_1 \} \subseteq h_{v_2 r_2}^{r_2}$
in~$\Gamma$.
Also, if~$r_2 \notin T_{v_1 r_1}^{r_1}$, then~$r_2 = v_1$.
\newcommand{\lemIndepAnglesContainIText}{}
\newcommand{\lemIndepAnglesContainIIText}{}
\begin{figure}
  \hfill
 \subfloat[]{\includegraphics[page=4, scale=1.0]{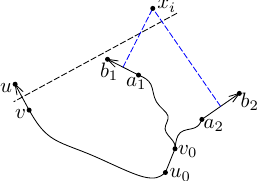}
  \label{fig:prelim:apex}}
   \hfill
 \subfloat[]{\includegraphics[page=6]{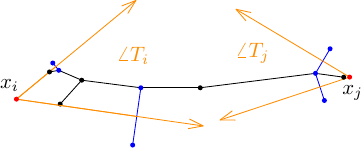}
  \label{fig:prelim:lem-2open-angles} }
  \hfill
 \subfloat[]{\includegraphics[page=7, scale=1.0]{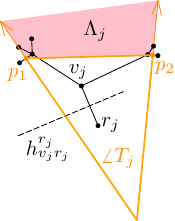}
  \label{fig:prelim:closed} }
  \hfill\null
  \caption{Proof of Lemma~\ref{lem:indep-angles-contain}.
  \protect\subref{fig:prelim:apex}~Proving $\apex(T_i) \in \polytope(T_j)$ using Lemma~\ref{lem:lem-slope}.
  \protect\subref{fig:prelim:lem-2open-angles}~Open angles of independent subtrees must contain apices of each other.
  \protect\subref{fig:prelim:closed}~Subtree~$T_j$ with closed angle~$\angle T_j$ and boundary segment~$p_1 p_2$. 
} 
\end{figure}
\begin{lemma}
  Let~$T_i$ and~$T_j$ be independent.
  \begin{compactitem}
    \item[(a)]  Let $|\angle T_i|$, $|\angle T_j|>0$ in~$\Gamma$.  Then, $\apex(\angle T_i) \in \angle T_j$ and~$\apex(\angle T_j) \in \angle T_i $.
    \item[(b)] Let~$|\angle T_i| > 0$, $|\angle T_j| \leq 0$, $p_1 p_2$ the boundary segment of~$\angle T_j$.  Then,~$\apex(T_i) \in \angle T_j$, and~$p_1, p_2 \in \angle T_i$.
    \item[(c)] Let~$|\angle T_i|$, $|\angle T_j| \leq 0$, $p_1 p_2$ the boundary segment of~$\angle T_j$.  Then, $p_1, p_2 \in \angle T_i$.
  \end{compactitem}
 \label{lem:indep-angles-contain}
\end{lemma}
\begin{proof}
 (a)~If~$\apex(T_i)$ is a node of~$T_i$, then~$\apex(T_i) \in T_i \subseteq \polytope(T_j) \subseteq \angle T_j$.
 Now let~$x_i = \apex(T_i)$ be the intersection of $\axis(e_1)$ and~$\axis(e_2)$,
 $e_k = a_k b_k$, $k=1,2$ the extremal edges of~$T_i$; see Fig.~\ref{fig:prelim:apex}. 
 Let~$\rho_k$ be the~$r_i$-$b_k$-path in~$T_i$ and~$u_0 \rightarrow v_0$ the last common edge of~$\rho_1$ and~$\rho_2$.
 Without loss of generality, let~$\vec{a_1 b_1}$ point upwards to the left and~$\vec{a_2 b_2}$ upwards to the right.
 Then,~$x_i$ must lie to the right of~$\vec{a_1 b_1}$ and to the left of~$\vec{a_2 b_2}$.
 Consider an edge~$u v$ in~$T_{u_0 v_0}^{u_0} + \{ u_0 v_0 \} $ such that~$v, u_0$ lie on the~$u$-$v_0$-path in~$T$.
 It is~$v_0, a_k, b_k \in \hp{uv}{v}$, $k=1,2$, and~$\axis(uv)$ does not cross
 the $v_0$-$b_k$-path. 
 Assume $x_i \notin \hp{uv}{v}$.
 Then, $\axis(uv)$ must cross both blue segments in Fig.~\ref{fig:prelim:apex},
 and for one~$k \in \{ 1,2\}$,
 rays~$\ray{v}{vu}$ and~$\ray{a_k}{a_k b_k}$ must be parallel or converge,
 a contradiction to Lemma~\ref{lem:lem-slope}.
 Hence, $x_i \in \hp{v u}{v}$. 
 Thus,~$x_i \in \polytope(T_j) \subseteq \angle T_j$.
 \par
(b)~For the cone~$C_j$
 from the definition of~$\angle T_j$, it must hold~$r_j \in C_j$.
  Let~$\Lambda_j = C_j \setminus \angle T_j$;
 see Fig.~\ref{fig:prelim:closed}.
 Let~$p_1 p_2$ be horizontal and~$C_j$ point downwards. 
 Since~$p_1, p_2$ lie on $T_j \setminus \{ r_j v_j\} \subseteq h_{v_j r_j}^{v_j}$,
 $\axis(v_j r_j)$ must cross the sides of~$C_j$ below~$p_1$ and~$p_2$.
 By Lemma~\ref{lem:lem-slope},~$\axis(v_j r_j)$ cannot cut the two unbounded sides of~$\Lambda_j$, since rays~$\ray{v_j}{v_j r_j}$ and~$\ray{a_k}{a_k b_k}$ diverge, $k=1,2$.
 Therefore, it is~$h_{v_j r_j}^{r_j} \cap \Lambda_j = \emptyset$.
 Hence,~$T_i\setminus \{ r_i\} \subseteq T_{r_j v_j}^{r_j}$ must be contained in~$\angle T_j$.
 By the same argument as in the proof of~(b), $\apex(T_i)$ lies in~$\polytope(T_j)$,
 and by choosing $vu = r_j v_j$ in~(a) we get~$\apex(T_i) \in \hp{v_j r_j}{r_j}$.
 Therefore, $x_i = \apex(T_i) \notin \Lambda_j$, and~$x_i \in \angle T_j$.
 Finally, since~$p_1, p_2$ lie on~$T_j$,
 it holds: $p_1, p_2 \in \polytope(T_i) \subseteq \angle T_i$.
  \par
 (c)~Since~$p_1, p_2$ lie on~$T_j$, it is~$p_1, p_2 \in \polytope(T_i)$.
 Also, $p_1, p_2 \in \hp{v_i r_i}{r_i}$.
 By the same argument as in the proof of~(a),
 $h_{r_i v_i}^{r_i} \cap \Lambda_i = \emptyset$.
 Hence, $p_1, p_2 \in \angle T_i$. 
\end{proof}
\begin{lemma}[generalization of Claim~4 in~\cite{lm-srgems-10}]
 Let~$T_i$, $T_j$ be two independent subtrees. Then, either~$|\angle T_i| > 0$
 or~$|\angle T_j| > 0$.
 \label{lem:two-closed-angles}
\end{lemma}
\begin{proof}
 Assume $|\angle T_i|,|\angle T_j| \leq 0$. By Lemma~\ref{lem:indep-angles-contain}(c), $\angle T_i$ contains the boundary segment~$p_j^1 p_j^2$ of~$\angle T_j$, and vice versa. This is not possible.
\end{proof}
We shall use the following lemma to provide a certificate of non-existence of a greedy drawing.
\begin{lemma}
Let~$T_i$, $i=1, \ldots, d$ be pairwise independent subtrees, and~$\alpha_i = |\angle T_i|$.  Then,  $$\sum_{i=1,\ldots,d, \alpha_i > 0} \alpha_i > (d-2)180\dg.$$
 \label{lem:indep-angles-sum}
\end{lemma}
\begin{proof}
 First, let~$\alpha_i > 0$, $i=1,\ldots,d$.
 Arranging all angles~$\angle T_i$ in accordance with Lemma~\ref{lem:indep-angles-contain}
 forms a convex~$2d$-gon, in which each~$\angle T_i$ provides one angle of size~$\alpha_i$
 and the remaining~$d$ angles are less than~$180\dg$ each.
 Then,~$d\cdot 180\dg + \sum_{i=1}^d \alpha_i > (2d-2) 180\dg$,
 and~$\sum_{i=1}^d \alpha_i > (d-2) 180\dg$. \par
 Now let~$\alpha_1 < 0$ and~$\alpha_i > 0$, $i=2,\ldots,d$.
 Then, arranging all angles~$\angle T_i$ in accordance with Lemma~\ref{lem:indep-angles-contain}
 forms a convex~$2d+1$-gon, in which~$\angle(T_1)$ provides
 two angles with sum~$180\dg$ or less,
  each~$\angle T_i$, $i=2,\ldots,d$ provides one angle of size~$\alpha_i$
  and the remaining~$d$ angles are less than~$180\dg$ each.
  Then,~$(d+1)\cdot 180\dg + \sum_{i=2}^d \alpha_i > (2d-1) 180\dg$,
 and~$\sum_{i=2}^d \alpha_i > (d-2) 180\dg$.
\end{proof}
Let~$T$ contain a set of~$n_k$ irreducible $k$-tuples, $k=3,4,5$, %
that are all pairwise independent.
Leighton and Moitra~\cite{lm-srgems-10} showed that for~$n_3\geq 6$ no greedy drawing of~$T$ exists.
We generalize this result slightly:
\begin{lemma}
 No greedy drawing of~$T$ exists if~$n_3 + 2 n_4 + 3 n_5 \geq 6$.
 \label{lem:indep-tuples-num}
\end{lemma}
\begin{proof}
 A triple has opening angle less than~$120\dg$, a quadruple less than~$60\dg$ and a quintuple cannot be drawn with an open angle.  Thus, the sum of all positive opening angles is at most $ 120\dg \cdot n_3 + 60\dg \cdot  n_4$. 
From the lemma's assumption we have $ - n_3 - 2 n_4 - 3 n_5 \leq -6 $.
By adding $3 n_3 + 3 n_4 + 3 n_5$ to both sides, we acquire $ 2 n_3 + n_4 \leq 3 n_3 + 3 n_4 + 3 n_5 - 6$.  Finally, multiplying by $60\dg$ provides $120\dg \cdot n_3 + 60\dg \cdot  n_4 \leq (n_3 + n_4 + n_5 - 2) 180\dg$, a contradiction to Lemma~\ref{lem:indep-angles-sum}.
\end{proof}
\subsubsection{Outline of the characterization.} 
Consider a node~$r \in V$ with neighbors~$v_1, \ldots, v_d$. The subtrees $T_i = T_{r v_i}^{v_i}+ r v_i$ with the common root $r$ are pairwise independent, $i = 1, \dots, d$. Consider angles $\varphi_i \geq 0$, such that $|\angle T_i| \leq \varphi_i$ in each greedy drawing of $T_i$. If either there exist $i,j \in \{ 1, \dots, d\}$, $i\neq j$, such that $\varphi_i, \varphi_j = 0$, or $\sum_{i=1}^d \varphi_i \leq (d-2) 180\dg$, then, by Lemmas~\ref{lem:two-closed-angles} and~\ref{lem:indep-angles-sum}, the tree~$T$ has no greedy drawing.\par
Determining tight upper bounds on~$|\angle T_i|$  will let us derive a sufficient condition for the existence of a greedy drawing of~$T$. Using the next result, we shall be able to compute such a bound for any rooted subtree in Section~\ref{sec:opening-angles}.
\subsection{Shrinking lemma}
We now present a lemma which is crucial for later proofs. 
Let the \emph{bounding cone} of a subtree $T_{rv}^v +rv$ defined for an edge $rv$ in a greedy drawing $\Gamma$ of~$T$ be the cone with apex $v$ %
and boundary rays~$\ray{v}{a_1 b_1}$ and~$\ray{v}{a_2 b_2}$ for extremal edges~$a_1 b_1$, $a_2 b_2$ of~$T_{rv}^v +rv$ that contains the drawing of $T_{rv}^v$.
\begin{figure}[tb]
 \centering
 \subfloat[]{
  \includegraphics[scale = 1.0, page=1]{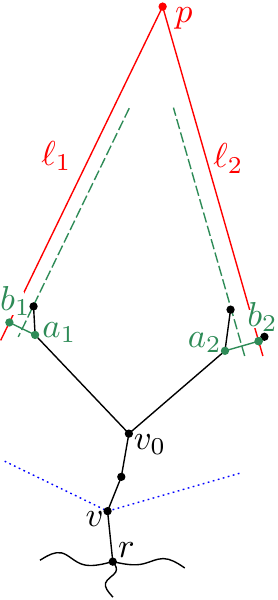} \label{fig:lem-shrink:1} }
  \hfill
  \subfloat[]{
  \includegraphics[scale = 1.0, page=6]{fig/full/lem-shrink.pdf} \label{fig:lem-shrink:2} } 
  \hfill
  \subfloat[]{
  \includegraphics[scale = 1.0, page=7]{fig/full/lem-shrink.pdf} \label{fig:lem-shrink:7} }
 \caption{Illustration of Lemma~\ref{lem:shrink}.
 \protect\subref{fig:lem-shrink:1}~Greedy drawing~$\Gamma$. Edges~$a_1 b_1$, $a_2 b_2$ are extremal.
 Dotted blue: bounding cone of~$T'$.
  \protect\subref{fig:lem-shrink:2}~Greedy drawing~$\Gamma'$. Subtree~$T_{rv}^v$ has been moved to
  a new point~$p \notin V$ and drawn infinitesimally small.
  \protect\subref{fig:lem-shrink:7} Drawings~$\Gamma$ and~$\Gamma'$ for the case when $a_1$,$b_1$ lie on the 
  $r$-$b_2$-path. Here, $p = b_2 \in V$. }
\label{fig:lem-shrink-1}
\end{figure}
\begin{lemma}
 Let~$T=(V,E)$ be a tree and~$T' = T_{rv}^v + rv$, $rv \in E$, a subtree of~$T$.
 Let~$\Gamma$ be a greedy drawing of~$T$, such that~$|\angle T'| > 0$.
 Then, there exists a point~$p$ in the bounding cone of~$T_{rv}^v$, such that
 shrinking~$T_{rv}^v$ infinitesimally and moving it to~$p$ keeps the drawing greedy, 
 and~$|\angle T'|$ remains the same.
 \label{lem:shrink}
\end{lemma}
\begin{proof}
Let~$e_i = a_i b_i$, $i=1,2$, be the two extremal edges of~$T'$ in~$\Gamma$, 
$\rho_i$ the $r$-$b_i$-path, and $a_i \in \rho_i$; see Fig.~\ref{fig:lem-shrink-1} for a sketch. We distinguish two cases: \par
 (1)~Edge~$e_1$ is not on~$\rho_2$ and edge~$e_2$ is not on~$\rho_1$.
 Then, $\{ a_1,b_1\} \subseteq \hp{a_2 b_2}{a_2}$,
 and $\{ a_2,b_2\} \subseteq \hp{a_1 b_1}{a_1}$.
 Let~$\ell_i$ be the line parallel to~$\axis(e_i)$ through~$b_i$ and~$p$ the intersection of~$\ell_1$ and~$\ell_2$;
 see Fig.~\ref{fig:lem-shrink:1}.
 Let~$v_0 \in V$ be the last common node of~$\rho_1$ and~$\rho_2$,
 and let~$\eta_i$ be the~$v_0$-$b_i$-path in~$T$, $i=1,2$. \par
 We now define three intermediate drawings.
 Let~$\Gamma_1$ be the drawing gained by replacing~$T'$ in~$\Gamma$ by
 the edge~$r v_0$ and the two paths~$\eta_1$ and~$\eta_2$,
 and let~$\Gamma_2 = \Gamma_1 - \eta_1 - \eta_2 + \{ v_0 b_1, v_0 b_2\}$; see Fig.~\ref{fig:lem-shrink:Gamma1}.
 By Lemma~\ref{lem:deg2-contract}, both~$\Gamma_1$ and~$\Gamma_2$ are greedy.
 Let~$\Gamma_3 = \Gamma_2 - \{ v_0 b_1, v_0 b_2\} + \{ v_0 p\}$.
 Let~$V_1$ be the node set of~$T_{vr}^r$ with addition of~$v_0$.
 Note that the nodes in~$V_1$ have the same coordinates in~$\Gamma$, $\Gamma_1$, $\Gamma_2$ and~$\Gamma_3$. 

\begin{figure}[tb]%
  \hfill
  \subfloat[$\Gamma_{1},\Gamma_{2},\Gamma_{3}$]{
  \includegraphics[scale = 1.0, page=10]{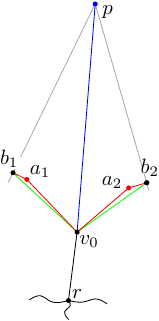} \label{fig:lem-shrink:Gamma1} 
  }
  \hfill
  \subfloat[]{
  \includegraphics[scale = 1.0, page=4]{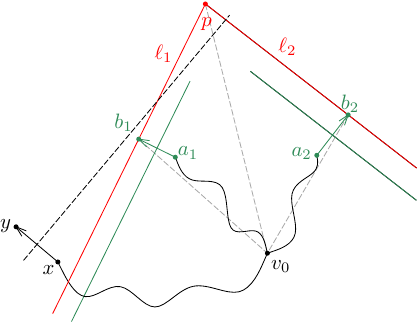} \label{fig:lem-shrink:raysDiverge} }
  \hfill\null
 \caption{
  Proof of Lemma~\ref{lem:shrink}.
  \protect\subref{fig:lem-shrink:Gamma1}:
  Intermediate drawings~$\Gamma_1$ (black and red), 
  $\Gamma_2$ (black and green) and~$\Gamma_3$ (black and blue). 
  \protect\subref{fig:lem-shrink:raysDiverge}: For an edge~$xy \notin T_{rv}^v$,
  its axis doesn't cross~$v_0 b_1$, $v_0 b_2$.
  It also doesn't cross~$v_0 p$ due to Lemma~\ref{lem:lem-slope}.
 }
 \label{fig:lem-shrink}
\end{figure}

 We have to prove the greediness of~$\Gamma_3$.
 Since~$p \notin V$, it doesn't follow directly from Lemma~\ref{lem:deg2-contract}.
 We first show that for an edge~$xy$ in~$\Gamma_3$, $xy \neq v_0 p$, where~$x$
 is closer to~$v_0$ in~$T$ than~$y$, it holds~$p \in \hp{xy}{x}$.
 Edge~$xy$ is also contained in~$\Gamma_1$.
 Nodes~$x$, $v_0$ and~$a_i$ lie on the $y$-$b_i$-path in~$T$, $i=1,2$.
 Hence,~$\{ v_0, a_1, a_2, b_1, b_2 \} \subseteq \eta_1 \cup \eta_2 \subseteq \hp{xy}{x}$,
 therefore,~$\axis(xy)$ doesn't cross edges~$v_0 b_1$, $v_0 b_2$.
 Now assume~$p \notin \hp{xy}{x}$.
 Then,~$\axis(xy)$ must cross~$v_0 p$, $b_1 p$
 and~$b_2 p$ (but not~$v_0 b_i$); see Fig.~\ref{fig:lem-shrink:raysDiverge}.
 This is only possible if for some $i \in \{1,2\}$,
 rays~$\ray{x}{xy}$ and $\ray{a_i}{a_i b_i}$ are parallel or converge,
 which is a contradiction to Lemma~\ref{lem:lem-slope}.
 \par
 Next, we show that~$V_1 \subseteq \hp{p v_0}{v_0}$.
 Without loss of generality, let~$v_0 b_1$ be directed upwards to the left
 and~$v_0 b_2$ upwards to the right.
 Note that~$a_1$ lies to the right of~$v_0 b_1$
 and~$a_2$ to the left of~$v_0 b_2$ 
 (otherwise, the edge~$a_i b_i$ would not be extremal in~$T'$).
Hence,~$\angle v_0 b_i p \geq 90\dg$.
Further, since~$v_0 \in \hp{a_i b_i}{a_i}$, it is~$\angle a_i b_i v_0 < 90\dg$, therefore, $\angle v_0 b_i p < 180\dg$, $i = 1,2$, and~$p$ lies inside the angle~$\angle b_1 v_0 b_2 < 180\dg$. 
 Let~$\Lambda$ 
 be the opening angle of the subtree induced by edges~$\{ r v_0, v_0 b_1, v_0 b_2 \}$ with root~$r$ in~$\Gamma_2$
 (blue in Fig.~\ref{fig:lem-shrink:v0p:1}).
 It is~$\Lambda \subseteq \hp{p v_0}{v_0}$ (see Lemma~\ref{lem:v0p}).
 Hence, $V_1 \subseteq \Lambda \subseteq \hp{p v_0}{v_0}$.
 This proves the greediness of~$\Gamma_3$. 
 Due to the extremality of~$a_1 b_1$, $a_2 b_2$, $p$ lies in the bounding cone of~$T'$.\par
 Removing~$v_0$ and connecting~$r$ to~$p$ keeps the drawing greedy.
 Finally, we acquire~$\Gamma'$ by drawing the subtree~$T_{rv}^v$ of~$T$ infinitesimally small at~$p$.
 Let~$C_1$ be the cone~$\angle T'$ in the original drawing~$\Gamma$,
 and~$C_2$ be the cone bounded by~$\ell_1$ and~$\ell_2$, $a_i \in C_2$.
 By Observation~\ref{lem:move-halfplane},
 $C_1 \subseteq C_2$.
 Consider an edge~$e$ in~$T_{rv}^v$, $e \notin \{ e_1, e_2\}$ in~$\Gamma$.
 Let~$\ell$ be the line parallel to~$\axis(e)$ through~$p$.
 Due to the extremality of~$e_1$, $e_2$, cone~$C_2$ lies on one side of~$\ell$.
 Therefore, since $V_1 \subseteq C_2$, the drawing~$\Gamma'$ is greedy,
 and it is~$\angle T' = C_2$.
 Since~$\ell_i$ is parallel to~$\axis(a_i b_i)$, 
 $|\angle T'|$ in~$\Gamma'$ is as big as in~$\Gamma$. \par
 (2)~Now assume $a_1 b_1$ lies on~$\rho_2$.
 Let~$\Gamma_4$ be the drawing obtained by replacing~$T'$ in~$\Gamma$ by edge~$r b_2$. By Lemma~\ref{lem:deg2-contract},~$\Gamma_4$ is greedy.
 It is~$b_2 \in \hp{a_1 b_1}{b_1}$.
 Similar to~(1), we acquire~$\Gamma'$ by drawing the subtree~$T_{rv}^v$ of~$T$ infinitesimally small at~$p = b_2$.
 Then, $|\angle T'|$ remains the same as in~$\Gamma$, see Fig.~\ref{fig:lem-shrink:7}.
\end{proof}
\begin{figure}[tb]
 \hfill
  \subfloat[]{
  \includegraphics[scale = 1.0, page=5]{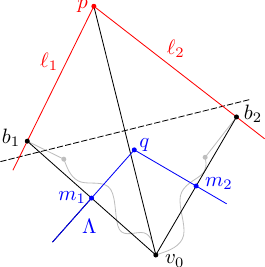} \label{fig:lem-shrink:v0p:1} }
 \hfill
   \subfloat[]{
  \includegraphics[scale = 1.0, page=11]{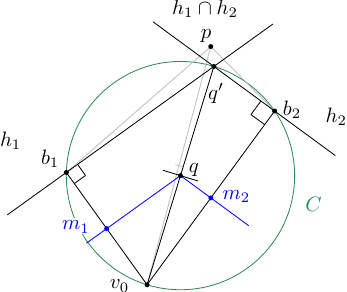}
\label{fig:lem-shrink:v0p} }\hfill\null 
  \caption{Proof of Lemma~\ref{lem:v0p}.
   $\axis(v_0 p)$ doesn't cross~$v_0 m_i$, and~$p$ doesn't lie inside the circle~$C$. Hence,~$\Lambda \subseteq \hp{v_0 p}{v_0}$.}
 \end{figure}
 \begin{lemma}
 Consider five points $v_0$, $b_1$, $b_2$, $p$, $q$, such that
 $p$ lies inside the angle $\angle b_1 v_0 b_2 < 180\dg$, $\angle v_0 b_1 p, \angle v_0 b_2 p \geq 90\dg$,
 $b_1 \in \hp{v_0 b_2}{v_0}$,  $b_2 \in \hp{v_0 b_1}{v_0}$.
 Let~$\Lambda$ be the cone bounded by~$axis(v_0 b_1)$ and~$\axis(v_0 b_2)$, $v_0 \in \Lambda$.
 Then, $\Lambda \subseteq \hp{v_0 p}{v_0}$.%
 \label{lem:v0p}
\end{lemma}
\begin{proof}
 Without loss of generality, let~$v_0 b_1$ be directed upwards to the left
 and~$v_0 b_2$ upwards to the right; see Fig.~\ref{fig:lem-shrink:v0p}.
 Let~$m_i$ be the midpoint of~$v_0 b_i$, $i=1,2$.
 Since~$\angle v_0 b_i p \geq 90\dg$, it is~$|v_0 b_i| \leq |v_0 p|$,
 and~$\axis(v_0 p)$ cannot cross the interior of~$v_0 m_i$. \par
 Let~$\{ q \} = \axis(v_0 b_1) \cap \axis(v_0 b_2)$.
 It remains to show that~$|v_0 q| \leq |q p|$.
 Let~$h_i$ be the closed half-plane bounded from below by the line orthogonal to~$v_0 b_i$
 through~$b_i$, $i=1,2$.
 Since~$\angle v_0 b_i p \geq 90\dg$, it must hold:~$p \in h_1 \cap h_2$.
 Let~$q'$ be the intersection point of the boundaries of~$h_1$ and~$h_2$.
 Let~$C$ be the circle with center~$q$ and radius~$v_0 q$.
 By elementary geometric arguments, points~$b_1$, $b_2$ and~$q'$ lie on the boundary of~$C$.
 Since a line can cross a circle at most twice,
 it is~$h_1 \cap h_2 \cap C = \{ q' \}$; see Fig.~\ref{fig:lem-shrink:v0p}, right.
 Hence,~$p$ cannot lie in the interior of~$C$.
\end{proof}
\section{Opening angles of rooted trees}
\label{sec:opening-angles}
The main idea of our decision algorithm is to process the nodes of~$T$ bottom-up while calculating tight upper bounds on the maximum possible opening angles of the considered subtrees.
If~$T$ contains a node of degree~5, it cannot be drawn with an open angle,
since each pair of consecutive edges forms an angle strictly greater than~$60\dg$.
In this section, we consider trees with maximum degree~4.
\par 
\begin{figure}[t]
 \hfill
  \subfloat[case I]{\includegraphics[page=12, scale = 0.95]{fig/full/combine-cases.pdf}
  \label{fig:combine-cases:0}
 } \hfill{}
 \subfloat[case II]{\includegraphics[page=1, scale = 0.95]{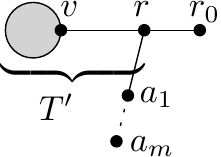}
  \label{fig:combine-cases:1}
 } \hfill{}
 \subfloat[case III]{\includegraphics[page=2, scale = 0.95]{fig/full/combine-cases.pdf}
  \label{fig:combine-cases:2}
 } \hfill{}
 \subfloat[case IV]{\includegraphics[page=3, scale = 0.95]{fig/full/combine-cases.pdf}
  \label{fig:combine-cases:3}
 } \hfill{}
 \subfloat[case V]{\includegraphics[page=4, scale = 0.95]{fig/full/combine-cases.pdf}
  \label{fig:combine-cases:4}
 } \hfill{}
 \subfloat[case VII]{\includegraphics[page=15, scale = 0.95]{fig/full/combine-cases.pdf}\hfill\null
  \label{fig:combine-cases:5}
 }
 \caption{
 \protect\subref{fig:combine-cases:0}--\protect\subref{fig:combine-cases:4}:
 Possible cases when combining subtrees to maintain an open angle. Subtrees~$T_1, T_2$ have opening angles~$\in (90\dg, 120\dg)$.
 In case~VII (\protect\subref{fig:combine-cases:5}) or in case VI
 ($|\angle T_i| \leq 90\dg$ in~IV or~V
 for one~$i \in \{1,2\}$) no open angle is possible.}
 \label{fig:combine-cases}
\end{figure}
If a subtree~$T'$ can be drawn with an open angle~$\varphi - \eps$ for any~$\eps > 0$, but not~$\varphi$, %
we say that it has opening angle~$\varphi\mn$ and write~$|\angle T'| = \varphi\mn$.
For example, a triple has opening angle~$120\mn$ and a quadruple~$60\mn$.
We call a subtree \emph{non-trivial} if it is not a single node or a simple path.
Figure~\ref{fig:combine-cases} shows possibilities to combine or extend non-trivial subtrees~$T', T_1, T_2$.
We shall now prove tight bounds on the possible opening angles for each construction.
As we  show later, 
only cases~I--V are feasible for the resulting subtree to have an open angle.
\begin{figure}
 \centering
 \subfloat[]{\includegraphics[page=5]{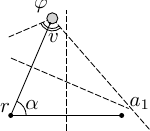}
 \label{fig:combine-cases:1proof:1}}
 \hfill
  \subfloat[]{\includegraphics[page=7]{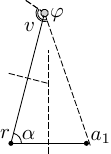}
 \label{fig:combine-cases:1proof:2}}
 \hfill
 \subfloat[]{\includegraphics[page=13]{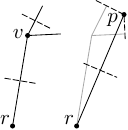}
 \label{fig:combine-cases:1proof:c2}}
 \hfill
   \subfloat[]{\includegraphics[page=8]{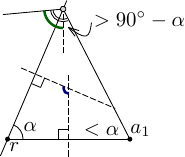}
 \label{fig:combine-cases:1proof:3}}
 \hfill
 \subfloat[]{
  \includegraphics[page=2]{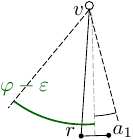}
  \label{fig:bintrees:combine1}
 }
 \caption{Optimal construction and tight upper bound for case~II.}
\end{figure}
\begin{lemma}
 Let~$T'$ be a subtree with positive opening angle. Consider the subtree~$\overline{T} = T' + r r_0$ in Fig.~\ref{fig:combine-cases:0}.
 Then~$\overline{T}$ has the same maximum opening angle as~$T'$.%
 \label{lem:combine-case0}
\end{lemma}
\begin{proof}
 Obviously,~$\overline{T}$ cannot have a bigger maximum opening angle than $T'$.
 By Lemma~\ref{lem:shrink}, for every greedy drawing of~$T'$ there exists a greedy drawing with an opening angle~$\angle T'$ of the same size 
 in which~${T'}_{rv}^{v}$ is drawn infinitesimally small.
 We then draw~$v,r,r_0$ collinearly inside~$\angle T'$.
\end{proof}
\begin{table}
 \centering
 \begin{tabular}{p{.07\linewidth} p{.13\linewidth} p{.13\linewidth} p{.13\linewidth} p{.42\linewidth} r}
  case & $\varphi_1$ & $\varphi_2$ & $\varphi_3$ & maximum $|\angle \overline{T}|$ & proof\\\hline
  I & $(0\dg, 180\dg]$ & - & - & ${\varphi_1}^-$ & Lem.~\ref{lem:combine-case0}\\
  II.i & $180\dg$ & $(90\dg, 120\dg]$ & - & $({\frac{\varphi_2}{2} + 45\dg)}^- \in (90\dg, 120\dg)$ & Lem.~\ref{lem:combine-case1} \\
  II.ii & $180\dg$ & $(0\dg, 60\dg]$ & - & ${\varphi_2}^- \in (0\dg, 60\dg)$ & Lem.~\ref{lem:combine-case1}\\
  III & $180\dg$ & $180\dg$ & $(0\dg, 120\dg]$ & ${\frac{\varphi_3}{2}}^- \in (0\dg, 60\dg)$ & Lem.~\ref{lem:combine-case2}\\
  IV & $(90\dg, 120\dg]$ & $(90\dg, 120\dg]$ & - & ${(\varphi_1 + \varphi_2 - 180\dg)}^- \in (0\dg, 60\dg)$ & Lem.~\ref{lem:combine-case3}\\
  V & $180\dg$ & $(90\dg, 120\dg]$ & $(90\dg, 120\dg]$ & $({\frac{3}{4}\varphi_2 + \frac{1}{2}\varphi_3 - 112.5\dg)}^- \in (0\dg, 60\dg)$ & Lem.~\ref{lem:combine-case4} \\
  VI & $(0\dg, 120\dg]$ & $(0\dg, 90\dg]$ & - & $<0\dg$ & Lem.~\ref{lem:combine:closed}\\
  VII & $(0\dg, 120\dg]$ & $(0\dg, 120\dg]$ & $(0\dg, 120\dg]$ & $<0\dg$ & Lem.~\ref{lem:combine:closed}\\\hline
 \end{tabular}
 \normalsize
 \caption{Computing maximum opening angle of the combined subtree~$\overline{T}$.
 Let~$|\angle T_i| = \varphi_i^-$, $\varphi_i \geq \varphi_{i+1}$, and $|\angle T_i| = \varphi_i = 180\dg$
 if $T_i$ is a path.}
 \label{tab:combine}
\end{table} 
To compute the maximum opening angle of the combined subtree~$\overline{T}$ 
in cases~II--V,
we use the following strategy.
We show that applying Lemma~\ref{lem:shrink} to~$T'$ does not decrease the opening angle of~$\overline{T}$ in a drawing.
Hence, it suffices to consider only drawings in which~${T'}_{r v}^{v}$
is shrunk to a point. We than obtain an upper bound by solving a linear maximization problem.
Finally, we construct a drawing with an almost-optimal opening angle for~$\overline{T}$
inductively using an almost-optimal construction for~$T'$.
Tight upper bounds on opening angles of the combined subtree~$\overline{T}$ for all possible cases
are listed in Table~\ref{tab:combine}.
Note that no bounds in~$(120\dg,180\dg)$ and~$(60\dg,90\dg]$ appear.
We now present the proofs for cases~II--V.
\begin{lemma}
 Let~$T'$ be a subtree with~$\angle T' = \varphi\mn$, and consider the subtree~$\overline{T} = T' + r r_0 + r a_1 + a_1 a_2 + \ldots + a_{m-1} a_m$ in Fig.~\ref{fig:combine-cases:1}.
 Then~$|\angle \overline{T}| = (45\dg + \frac{\varphi}{2})\mn$ if $\varphi > 90\dg$ (case~(i)), and~$|\angle \overline{T}| = \varphi\mn$ if $\varphi \leq 90\dg$ (case~(ii)).
 \label{lem:combine-case1}
\end{lemma}
\begin{proof}
  First, let~$m=1$. (i)~Consider a greedy drawing~$\Gamma$ of~$\overline{T}$.
  Let~$a_1 r$ be drawn horizontally and~$v$ above it and to the left of~$\axis(r a_1)$; 
  see Fig.~\ref{fig:combine-cases:1proof:1},\subref*{fig:combine-cases:1proof:2},\subref*{fig:combine-cases:1proof:3}.
  Due to Lemma~\ref{lem:lem-slope}, the right boundary of~$\angle \overline{T}$ is formed by~$\axis(r a_1)$.
  The left boundary is either formed by (1)~the left boundary of~$\angle T'$ 
  (see Fig.~\ref{fig:combine-cases:1proof:1}), 
  or (2)~by~$\axis(rv)$ (Fig.~\ref{fig:combine-cases:1proof:2}).
  We apply Lemma~\ref{lem:shrink} to~${T'}_{rv}^v$ in~$\Gamma$
  and acquire~$\Gamma'$, in which~${T'}_{rv}^{v}$ is drawn infinitesimally small.
  In~$\Gamma'$,~$\axis(r a_1)$ remains the right boundary of~$\angle \overline{T}$. 
  In case~(1), the left boundary of~$\angle \overline{T}$ is again formed by the left boundary of~$\angle T'$,
  and~$|\angle \overline{T}|$ remains the same.
  In case~(2), the subtree~${T'}_{rv}^{v}$
  must lie to the right of~$\vec{rv}$ in~$\Gamma$
  (since each edge in it is oriented clockwise relative to~$\vec{r v}$),
  and so does the point~$p$ from Lemma~\ref{lem:shrink}.
  Thus, the edge~$rv$ is turned clockwise in~$\Gamma'$,
  and~$|\angle \overline{T}|$ increases;
  see Fig.~\ref{fig:combine-cases:1proof:c2}.
  Thus, to acquire an upper bound for~$|\angle \overline{T}|$
  it suffices to only consider drawings in which~${T'}_{rv}^{v}$ is drawn infinitesimally small.
  Let~$\alpha = \angle a_1 r v$.
  Then, for~$\overline{\varphi} = |\angle \overline{T}|$ it holds:
  $ \overline{\varphi} \leq 180\dg - \alpha$, $\overline{\varphi} < \varphi - 90\dg + \alpha$;
  see the blue and green angles in Fig.~\ref{fig:combine-cases:1proof:3}.
  Thus,~$\overline{\varphi}$ lies on the graph~$f(\alpha) = 180\dg - \alpha$ or below it
  and strictly below the graph $g(\alpha) = \varphi - 90\dg + \alpha$.
  Maximizing over~$\alpha$ gives $\overline{\varphi} < 45\dg + \frac{\varphi}{2}$. %
  We can achieve~$\overline{\varphi} = (45\dg + \frac{\varphi}{2})\mn$ by choosing~$\alpha = 135\dg - \frac{\varphi}{2} + \eps'$ 
  and drawing~${T'}_{rv}^{v}$ sufficiently small 
  with~$|\angle T'| = \varphi - \eps$ for sufficiently small~$\eps, \eps' > 0$. \par
  (ii)~Obviously, $|\angle T'| \geq |\angle \overline{T}|$. 
  For the second part, see Fig.~\ref{fig:bintrees:combine1}.
  We choose $\angle a_1 r v = 90\dg - \frac{\eps}{2}$ and draw~$r a_1$%
  long enough, such that its axis doesn't cross~${T'}_{rv}^{v}$.
  We rotate~${T'}_{rv}^{v}$ such that the right side of the opening angle~$\angle T'$ and~$r v$ form an angle~$\frac{3\eps}{2}$.
  Then, the opening angle~$\varphi'$ of the drawing is defined by the left side of~$\angle T'$ and the axis of~$r a_1$ and is~$\varphi - \eps$. \par
  For~$m \geq 2$, draw~$a_2, \ldots, a_m$ collinear with~$r a_1$ and infinitesimally close to $a_1$.
\end{proof}
\begin{lemma}
 Let~$T'$ be a subtree with~$|\angle T'| = \varphi\mn < 120\dg$, and consider subtree~$\overline{T} = T' + \{ r r_0, r a_1, \ldots, a_{m-1} a_m, r b_1, \ldots, b_{k-1} b_k \}$ in Fig.~\ref{fig:combine-cases:2}.  Then,~$\angle \overline{T} = \frac{\varphi}{2}\mn$.
 \label{lem:combine-case2}
\end{lemma}
\begin{figure}[tb]
 \hfill
  \subfloat[]{\includegraphics[page=6]{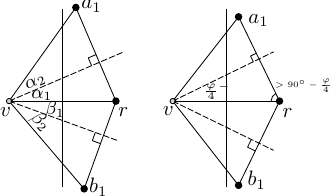}
 \label{fig:combine-cases:2proof:1}}
  \hfill
  \subfloat[]{\includegraphics[page=14]{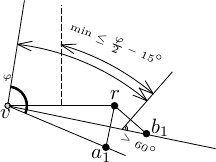}
 \label{fig:combine-cases:2proof:2}}
  \hfill\null
 \caption{Optimal construction and tight upper bound for case~2.}
\end{figure}
\begin{proof}
First, let~$k = m = 1$.
Consider a greedy drawing~$\Gamma$ of~$\overline{T}$ with~$|\angle \overline{T}|>0$.
Let~$rv$ be horizontal in~$\Gamma$ and let~$v$ lie to the left of~$r$.
There exist two possibilities for~$\Gamma$.
The edge~$rv$ can be either drawn inside the angle~$\angle a_1 r b_1 < 180\dg$ (see Fig.~\ref{fig:combine-cases:2proof:1})
or on the outside of it (see Fig.~\ref{fig:combine-cases:2proof:2}). \par
In the first case, let~$a_1$ lie above~$rv$ and~$b_1$ below.
Then, the upper boundary of~$\angle \overline{T}$ is formed by~$\axis(r a_1)$
and the lower by~$\axis(r b_1)$.
It remains to be the case after applying Lemma~\ref{lem:shrink} to~$T'$;
see Fig.~\ref{fig:combine-cases:2proof:1} for the corresponding drawing~$\Gamma'$.
We can assume that~$\angle v r a_1, \angle v r b_1 < 90\dg$ (otherwise,
we can increase~$\angle \overline{T}$ by turning~$r a_1$ counterclockwise or~$r b_1$ clockwise).
It must be~$\alpha_1 < \alpha_2$, $\beta_1 < \beta_2$ and~$\alpha_1 + \alpha_2 + \beta_1 + \beta_2 < \varphi$.
Thus, the opening angle in this construction is~$\alpha_1 + \beta_1 < \frac{\varphi}{2}$.
The angle~$|\angle \overline{T}| = \frac{\varphi}{2}\mn$ can be achieved by choosing
$\alpha_1 = \beta_1 = \frac{\varphi}{4} - 2\eps$, $\alpha_2 = \beta_2 = \frac{\varphi}{4} - \eps$
for a sufficiently small~$\eps$.
Then, $\angle a_1 v r < \angle a_1 r v $, and~$|r a_1| < |a_1 v|$
 (it is $\frac{\varphi}{2}^- < 60\dg \leq 90\dg - \frac{\varphi}{4}$).
Hence, the drawing is greedy and has opening angle~$|\angle \overline{T}| = \frac{\varphi}{2}\mn$.
\par
Now consider the second option for~$\Gamma$;
see Fig.~\ref{fig:combine-cases:2proof:2}.
Let~$r a_1$ be inside~$\angle v r b_1 < 180\dg$.
Then, $\axis(r b_1)$ is a boundary of the opening angle of~$\overline{T}$
and $\axis( r a_1 )$ is a boundary of the opening angle of~$\overline{T} - \{ r b_1 \}$.
By Lemma~\ref{lem:combine-case1}, $|\angle (\overline{T} - \{ r b_1 \})| < \frac{\varphi}{2} + 45\dg$,
and the drawing of $\overline{T}$ has opening angle at most
$ \frac{\varphi}{2} + 45\dg - 60\dg = \frac{\varphi}{2} - 15\dg$.
Hence, the first option is optimal. 
We add~$a_2, \ldots, a_m$ and~$b_2, \ldots, b_k$ similarly to the proof of Lemma~\ref{lem:combine-case1}.\par
Note that the new opening angle is $< 60\dg$.
\end{proof} 
\begin{figure}[tb]
 \hfill
 \subfloat[]{
  \includegraphics[page=1]{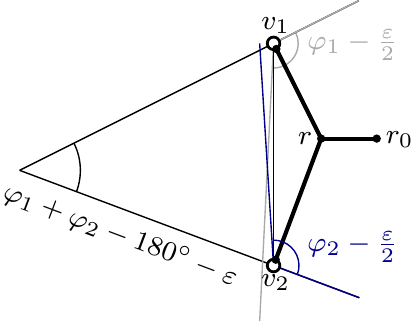}
  \label{fig:bintrees:combine2}
 }
 \hfill
 \subfloat[]{
 \includegraphics[page=9]{fig/full/combine-cases.pdf}
 \label{fig:combine-case4}}
  \hfill\null
 \caption{Sketch of the proof of Lemma~\ref{lem:combine-case3} and~\ref{lem:combine-case4}.}
\end{figure}
\begin{lemma}
 Let~$T_1$, $T_2$ be subtrees with~$|\angle T_i| = \varphi_i\mn \in (90\dg, 120\dg)$, $i = 1,2$, and consider subtree~$\overline{T} = T_1+ T_2 + \{ r r_0 \}$ in Fig.~\ref{fig:combine-cases:3}.
 Then,~$\angle \overline{T} = (\varphi_1 + \varphi_2 - 180\dg)\mn$.
 \label{lem:combine-case3}
\end{lemma}
\begin{proof}
 Let~$r r_0$ be drawn horizontally and~$v_1$ above~$v_2$ in~$\Gamma$.
 Then, the upper boundary of~$\angle \overline{T}$ is formed either by the upper boundary of~$\angle T_1$
 or by~$\axis(r v_1)$,
 and the lower boundary of~$\angle \overline{T}$ is formed either by the lower boundary of~$\angle T_2$
 or by~$\axis(r v_2)$.
 Since~$\apex(T_1) \in \angle T_2$ and~$\apex(T_2) \in \angle T_1$,
 an opening angle~$|\angle \overline{T}|  \geq \varphi_1 + \varphi_2 - 180\dg$ is not possible. \par
 For the lower bound, see the construction in Fig.~\ref{fig:bintrees:combine2}.
 Both~$({T_i})_{r v_i}^{v_i}$ are drawn infinitesimally small.
 The lower boundary ray of~$\angle T_1$ and the upper boundary ray of~$\angle T_2$
 have intersection angle~$\eps$, and the other two sides form an angle~$\varphi' = \varphi_1 + \varphi_2 - 180^\circ - \eps$.
 The edges~$r v_1$ and~$r v_2$ are drawn orthogonal to the 
 upper boundary of~$\angle T_1$ and lower boundary of~$\angle T_2$ respectively,
 so their axes are parallel to the boundary rays of~$\angle \overline{T}$.
 It is~$\angle v_1 r v_2 = 360\dg - \varphi_1 - \varphi_2 + \eps \geq 120\dg + \eps$.
 Hence, no axis crosses another edge, and the drawing is greedy.
 Note that the new opening angle is $< 60\dg$.
\end{proof}
\begin{lemma}
Let~$T_1$, $T_2$ be subtrees with~$\angle T_i = \varphi_i\mn \in (90\dg, 120\dg)$, $\varphi_1 \geq \varphi_2$, and consider subtree~$\overline{T} = T_1+ T_2 + \{ r r_0, r a_1, \ldots, a_{m-1} a_m \}$ in Fig.~\ref{fig:combine-cases:4}.
 Then,~$|\angle \overline{T}| = (\frac{3}{4} \varphi_1 + \frac{1}{2} \varphi_2 - 112.5\dg)\mn$.
 \label{lem:combine-case4}
\end{lemma}
\begin{proof}
 First, let~$m=1$.
 There exist two possibilities for a greedy drawing~$\Gamma$ of~$\overline{T}$.
 Edge~$r a_1$ can be either drawn inside the angle~$\angle v_1 r v_2 < 180\dg$ or outside it.
 For the first case, let~$v_1$ be above~$r a_1$ and~$v_2$ below it.
 Let~$\overline{T_1} = \overline{T} - T_2$ and~$\overline{T_2} = \overline{T} - T_1$.
 Then~$\axis(r a_1)$ forms the lower boundary of~$\angle \overline{T_1}$ and the upper boundary of~$\angle \overline{T_2}$.
 By Lemma~\ref{lem:combine-case1},~$|\angle \overline{T_1}| < \frac{\varphi_1}{2} + 45\dg$
 and~$|\angle \overline{T_2}| < \frac{\varphi_2}{2} + 45\dg$.
 Hence,~$|\angle \overline{T}| < \frac{\varphi_1}{2} + \frac{\varphi_2}{2} - 90\dg$. 
\par
 We now consider the second option.
 Let~$v_2$ be below~$v_1$ and~$a_1$ below~$v_2$ in~$\Gamma$; see Fig.~\ref{fig:combine-case4}.
 The upper boundary of~$\angle \overline{T}$ is either formed by the upper boundary of~$\angle T_1$ or by~$\axis r v_1$.
 The lower boundary of~$\angle \overline{T}$ is formed by~$\axis(r a_1)$. 
 Again, we acquire~$\Gamma'$ by applying~Lemma~\ref{lem:shrink} to~$T_1$ and then to~$T_2$.
 In~$\Gamma'$, both~$({T_i})_{r v_i}^{v_i}$ are drawn infinitesimally small.
 By a similar argument as in the proof of Lemma~\ref{lem:combine-case1},
 $|\angle \overline{T}|$ in~$\Gamma'$ is at least as big as in~$\Gamma$.
 Thus, for an upper bound it suffices to consider only greedy drawings 
 in which~$(T_1)_{r v_1}^{v_1}$ and~$(T_2)_{r v_2}^{v_2}$ are drawn infinitesimally small;
 see Fig.~\ref{fig:combine-case4} for one such drawing.
 Let~$\alpha = \angle v_1 v_2 r$, $\alpha_1 = \angle v_1 r v_2$
 and~$\gamma$ the angle formed by the upper boundary of~$\angle T_1$ and~$v_1 v_2$.
 It must be~$\alpha < \alpha_1$ and~$\gamma < \varphi_1$. %
 Then, for~$\varphi' = |\angle \overline{T}|$ it must hold:
 \begin{align*}
  & \varphi' < (90\dg - \alpha_1) + \frac{\varphi_2 - \alpha}{2} < 90\dg - \frac{3}{2}\alpha + \frac{\varphi_2}{2} =: f(\alpha), \\
  & \varphi' < (\alpha + \varphi_1 - 180\dg) + \frac{\varphi_2 - \alpha}{2} = \frac{\alpha}{2} - 180\dg + \varphi_1 + \frac{\varphi_2}{2} =: g(\alpha).
 \end{align*}
 Hence,~$\varphi' < \varphi_{\max} := \max_\alpha \min \{ f(\alpha), g(\alpha) \} = \frac{3}{4}\varphi_1 + \frac{1}{2}\varphi_2 - 112.5\dg$.
 We can achieve $\varphi_{\max}^-$ by choosing~$\alpha = 135\dg - \frac{\varphi_1}{2}$, 
 $\alpha_1 = \alpha + \eps$, $\gamma = \varphi_1 - \eps$ and~$\angle v_1 v_2 a_1 = \varphi_2 - \eps$.
 In this construction, $\axis(r v_1)$ is parallel to the upper boundary ray of~$\angle T_1$ (dashed green in Fig.~\ref{fig:combine-case4}).
 \par
 Since we assumed~$\varphi_1 > 90\dg$, it is $\frac{\varphi_1}{4}  > 22.5\dg$, and the second embedding option provides a bigger opening angle.
 Note that the new opening angle is $< 37.5\dg$.
\end{proof}
\begin{figure}[tb]
\hfill
 \subfloat[case VI]{\includegraphics[page=10]{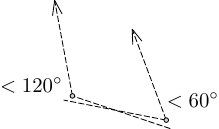}
 \label{fig:combine-closed:1}}\hfill
  \subfloat[case VII]{\includegraphics[page=11]{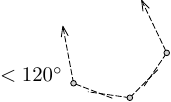}
 \label{fig:combine-closed:2}}\hfill\null
 \caption{In cases VI and VII, no open angle is possible.}
\end{figure}
\begin{lemma}
If either \textnormal{(VI)}~$|\angle T_i| < 90\dg$ in Fig.~\ref{fig:combine-cases:3} or~\ref{fig:combine-cases:4} for some~$i \in \{1,2\}$ or \textnormal{(VII)}~$|\angle T_i| < 120\dg$ for each~$i=1,\ldots,3$ in Fig.~\ref{fig:combine-cases:5}, it is~$|\angle \overline{T}| < 0$.
 \label{lem:combine:closed}
\end{lemma}
\begin{proof}
 First, let~$|\angle T_1| < 120\dg$ and~$|\angle T_2| \leq 90\dg$
 in Fig.~\ref{fig:combine-cases:3} or \ref{fig:combine-cases:4}.
 Since no tight upper bounds in range~$(60\dg, 90\dg]$ appear (see Table~\ref{tab:combine}),
 it is~$|\angle T_2| < 60\dg$. It must be~$\apex(T_i) \in \angle T_j$ for~$i \neq j$,
 therefore, no open angle is possible; see Fig.~\ref{fig:combine-closed:1}.
  Same holds for~$|\angle T_1|, |\angle T_2|, |\angle T_3| < 120\dg$; see Fig.~\ref{fig:combine-closed:2}.
\end{proof}
\section{Arranging rooted subtrees with open angles}
\label{sec:arrange-angles}
In this section, we consider the task of constructing a greedy drawing~$\Gamma$ of~$T$
by combining independent rooted subtrees with a common root.
The following problem (restricted to $n \in \{3,4,5\}$) turns out to be fundamental in this context.
\par

\begin{problem}\label{prob:ngon}
	Given $n$ angles~$\varphi_0$, \ldots, $\varphi_{n-1} > 0\dg$, is there a convex $n$-gon $P$ with corners $v_0, \dots, v_{n-1}$ (in arbitrary order) with interior angles $\psi_i < \varphi_i$ for $i=0, \dots, n-1$, such that the star $K_{1,n}$ has a greedy drawing with root $r$ inside $P$ and leaves $v_0, \dots, v_{n-1}$?
\end{problem}
If Problem~\ref{prob:ngon} has a solution we write $\{\varphi_0, \ldots, \varphi_{n-1}\} \in \gP{n}$.
It can be solved using a series of following optimization problems 
(one for each fixed cyclic ordering of~$(\varphi_1, \ldots, \varphi_n)$).
\begin{equation}
\begin{aligned}
 & \textnormal{maximize } \eps \textnormal{ under: }
 \eps,\alpha_i,\beta_i,\gamma_i \in [0\dg, 180\dg], \; i=0, \ldots, n-1 \\
 &\beta_i + \eps \leq \alpha_i, \;\; \gamma_i + \eps \leq \alpha_i, \; \; \beta_i + \gamma_{i+1} + \eps \leq \varphi_i \;\; (i \textnormal{ mod } n) \\
 &\alpha_i + \beta_i + \gamma_i = 180\dg, \;\; \alpha_0 + \ldots + \alpha_{n-1} = 360\dg\\ %
 &\sin(\beta_0) \cdot \ldots \cdot \sin(\beta_{n-1}) = \sin(\gamma_0) \cdot \ldots \cdot \sin(\gamma_{n-1})
\end{aligned}
\tag{*}\label{eq:opt-problem}
\end{equation}
The last constraint in~\optprobref follows from applying the law of sines and is known as the wheel condition 
in the work of di Battista and Vismara~\cite{DiBattistaVismara1993}.
\begin{lemma}
It is~$\{ \varphi_0, \ldots, \varphi_{n-1} \} \in \gP{n}$ if and only if there exists a solution of~\optprobref with~$\eps > 0$ for an ordering~$( \varphi_0, \ldots, \varphi_{n-1} )$.
\label{lem:Pn-Opt-equiv}
\end{lemma}
\begin{proof}
A solution to Problem~\ref{prob:ngon} provides a solution to~\optprobref
by the construction in Fig.~\ref{fig:opt-problem},
since~$|r v_{i-1}|, |r v_i| < |v_{i-1} v_i| \Leftrightarrow \beta_i, \gamma_i < \alpha_i$.
Conversely, a solution to~\optprobref
provides a greedy drawing of~$K_{1,n}$.
Without loss of generality, consider nodes~$v_0$, $v_2$,
and let the counterclockwise turn from~$r v_0$ to~$r v_2$ be at most~$180\dg$;
see Fig.~\ref{fig:PnGreedy-equiv}.
Then, $\angle v_0 v_2 r \leq \beta_2 < \alpha_1 + \alpha_2$,
 $\angle v_2 v_0 r \leq \gamma_1 < \alpha_1 + \alpha_2$, hence, $|r v_0|, |r v_2| < |v_0 v_2|$.
 \end{proof}
\begin{figure}[tb]
 \centering
 \hfill
 \subfloat[]{
 \includegraphics[scale=0.95, page=1]{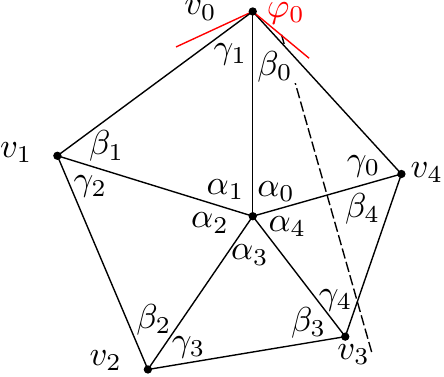}
 \label{fig:opt-problem}}
\hfill
\subfloat{\includegraphics[scale=0.95, page=6]{fig/full/opt-problem.pdf}
  \label{fig:PnGreedy-equiv}
}\hfill
  \subfloat[]{\includegraphics[scale=0.95, page=2]{fig/full/opt-problem.pdf}
  \label{fig:PnGreedy}
 }\hfill\null
 \caption{
  \protect\subref{fig:opt-problem}~Sketch for the optimization problem~\optprobref.
  \protect\subref{fig:PnGreedy-equiv}~In a drawing of~$K_{1,n}$ induced by a solution of~\optprobref, path~$v_i,r,v_j$ is distance-decreasing.
  \protect\subref{fig:PnGreedy} Solving~\optprobref lets us construct greedy drawings by placing sufficiently small drawings of subtrees at $n$-gon corners.
  }
\end{figure}
Deciding whether a solution of~\optprobref with~$\eps>0$ exists
is in fact equivalent to deciding whether the wheel condition can be satisfied in the interior of a
$2n-1$-dimensional simplex.
\begin{observation}
 Let~$n=3,4,5$, $\varphi_i \in [0\dg, 180\dg]$, $i=0,\ldots,n-1$, such that 
 $\sum_{i=0}^{n-1} \varphi_i > (n-2)180\dg$.
 For a permutation~$\tau$ of~$(0,\ldots,n-1)$, define a $2n-1$-dimensional simplex~$S_\tau$ as follows:
 \begin{align*}
  S_\tau   = \left\{(\beta_0,\ldots,\beta_{n-1},\gamma_0,\ldots,\gamma_{n-1}) \;\; \middle| \;\;
     \begin{aligned}  
  		 & \textnormal{for } i=0,\ldots,n-1: \\
  		 & \beta_i \geq 0, \gamma_i \geq 0, \\
  		 & \beta_i + \gamma_{i+1} \leq \varphi_{\tau(i)}, \\
  		 & 180\dg - \beta_i - \gamma_i \geq \beta_i, \\
  		 & 180\dg - \beta_i - \gamma_i \geq \gamma_i; \\
  		 & \sum_{i=0}^{n-1} (\beta_i + \gamma_i) = (n-2)180\dg
  \end{aligned}	  
  		\right\}.
\end{align*}
Define
$$ \omega(\beta_0,\ldots,\beta_{n-1},\gamma_0,\ldots,\gamma_{n-1}) = \prod_{i=0}^{n-1} \sin(\beta_i) - \prod_{i=0}^{n-1} \sin(\gamma_i). $$
Then, $\{ \varphi_0, \ldots, \varphi_{n-1}\} \in \gP{n}$ if and only if the function~$\omega$
has a zero in the interior of the simplex~$S_\tau$ for some permutation~$\tau$.
\label{lem:w-zero-on-simplex}
\end{observation}
\begin{theorem}
 For~$n = 3,4,5$, consider trees~$T_i$, $i=0, \ldots, n-1$ with root~$r$, edge~$r v_i$ in~$T_i$,
 $\deg(r) = 1$ in~$T_i$, $T_i \cap T_j = \{ r \}$ for~$i \neq j$, such that each~$T_i$ has a drawing 
 with opening angle at least~$0 < \varphi_i - \eps  < 180\dg$ for any~$\eps > 0$.
 Then, tree~$T = \bigcup_{i=0}^{n-1} T_i$ has a greedy drawing with
 $|\angle T_i| < \varphi_i$ for all~$i=0,\ldots,n-1$
 if and only if~$\{ \varphi_0, \ldots, \varphi_{n-1} \} \in \gP{n}$.
 \label{thm:p5equiv}
\end{theorem}
\begin{proof}
 First, consider a drawing of~$K_{1,n}$ with edges~$r v_i$ that solves~$\gP{n}$,
 and, without loss of generality, let the angles be ordered such that~$\psi_i := \angle v_{i-1} v_i v_{i+1} < \varphi_i$.
 We create a greedy drawing~$\Gamma$ of~$T$ by drawing~$({T_i})_{r v_i}^{v_i}$ infinitesimally small at~$v_i$
 with opening angle~$\varphi_i - \eps > \psi_i$ for a sufficiently small~$\eps > 0$
 and orienting it such that~$v_j \in \angle T_i$ for all~$j \neq i$; see Fig.~\ref{fig:PnGreedy}. \par
 Now assume a greedy drawing~$\Gamma_0$ of~$T$ with $|\angle T_i| < \varphi_i$, $i=0,\ldots,n-1$ exists.
 For one~$i$, it might be~$|\angle T_i| < 0$ in~$\Gamma_0$.
 Then, there also exists a greedy drawing~$\Gamma$, in which
 $0 < |\angle T_j| < \varphi_j$, $j=0,\ldots,n-1$.
 By Lemma~\ref{lem:two-closed-angles},
 the subtree $\overline{T} = \{ r v_i\} + \bigcup_{j \neq i} T_j$ must have an open angle in~$\Gamma_0$.
 We then obtain~$\Gamma$ by making the edge~$r v_i$ sufficiently long inside~$\angle \overline{T}$
 and drawing~$T_i$ with~$|\angle T_i|>0$,
 such that $\overline{T}  \subseteq \angle T_i$
 and $T_i \subseteq  \angle \overline{T}$.
 \par
 We apply Lemma~\ref{lem:shrink} to~$T_0$, then to~$T_1$, \ldots, $T_{n-1}$
 and obtain a greedy drawing~$\Gamma'$ of~$T$ with opening angles~$\angle T_i$ of same size,
 such that each subtree~$({T_i})_{r v_i}^{v_i}$ is drawn infinitesimally small at~$v_i$.
 For $n=4,5$, for each pair of consecutive edges~$r v_i$, $r v_j$ in~$\Gamma'$
 the turn from~$r v_i$ to $r v_j$ is less than~$180\dg$, %
 so $r$ lies inside the convex polygon with corners $v_0, \ldots, v_{n-1}$.  Therefore, $\Gamma'$ directly provides a solution of~$\gP{n}$.
 For~$n=3$, $v_1$ might lie inside angle~$\angle v_0 r v_2 \leq 180\dg$.
 However, since~$\varphi_0 + \varphi_1 + \varphi_2 > 180\dg$,
 it is~$\{ \varphi_0 , \varphi_1 , \varphi_2 \} \in \gP{3}$; see Lemma~\ref{lem:greedy-K13}. 
\end{proof}
Although Problem~\optprobref is non-linear,
we are almost always able to give tight conditions for the existence of the solution; 
see Table~\ref{tab:P}, which summarizes Lemmas~\ref{lem:greedy-K13}
to~\ref{lem:2x180}.
\begin{table}
 \centering
 \begin{tabular}{l p{.49\linewidth} p{.33\linewidth} r}
  $n$ & case & $\{ \varphi_0, \ldots, \varphi_{n-1} \} \in \gP{n}$ iff & proof \\\hline
  $3,4$ & & always & Lem.~\ref{lem:greedy-K13}, \ref{lem:greedy-K14} \\
  $5$ & $\varphi_0 = \ldots = \varphi_3 = 180\dg$ & always & Lem.~\ref{lem:4x180} \\
  $5$ & $\varphi_0 \leq 120\dg$ & always & Lem.~\ref{lem:P5-120} \\
  $5$ & $\varphi_0 = \ldots = \varphi_2 = 180\dg$ & $\varphi_3 + \varphi_4 > 120\dg$ & Lem.~\ref{lem:P5-3x180} \\
  $5$ & $\varphi_0 = \varphi_1 = 180\dg$ & $\varphi_2 + \varphi_3 + \varphi_4 > 240\dg$ & Lem.~\ref{lem:2x180} \\
  $5$ & $\varphi_0 = 180\dg, \varphi_1, \varphi_2, \varphi_3 \in (90\dg, 120\dg], \varphi_4 \leq 60\dg$ & \centering ? & \\
  $5$ & $\varphi_0 = 180\dg, \varphi_1, \ldots, \varphi_4 \in (90\dg, 120\dg]$ & \centering ? & \\\hline
 \end{tabular}
 \normalsize
 \caption{Solving non-linear problem~\gP{n} explicitly.
 Let~$\varphi_i \geq \varphi_{i+1}$, 
 $\varphi_i \in (0\dg,60\dg] \cup (90\dg,120\dg] \cup \{ 180\dg \}$,
 $\sum_{i=0}^{n-1} \varphi_i > (n-2)180\dg$.}
 \label{tab:P}
\end{table} 
\begin{lemma}
 For angles~$\varphi_0$, $\varphi_1$, $\varphi_2>0$, $\sum_{i=0}^2 \varphi_i > 180\dg$, it holds: $ \{ \varphi_0, \varphi_1, \varphi_2\} \in \mathcal P^3$.
 \label{lem:greedy-K13}
\end{lemma}  
\begin{proof}
 It is possible to choose~$0 < \psi_i < \varphi_i$, $i=0,\ldots,2$,
 such that~$\sum_{i=0}^2 \psi_i = 180\dg$.
 In~Problem~\optprobref, we set
 $\beta_i = \gamma_{i+1} = \frac{\psi_i}{2}$; see Fig.~\ref{fig:lemP3}.
 It is $\beta_i + \gamma_i = \frac{1}{2}(\psi_i + \psi_{i-1}) < 90\dg < \alpha_i$.
 Therefore, this angle assignment
 satisfies the constraints in Problem~\optprobref for some positive~$\eps$.
\end{proof}
\begin{lemma}
 For~$n = 4,5$ and angles~$\varphi_0, \ldots, \varphi_{n-1} \leq 120\dg$, $\sum_{i=0}^{n-1} \varphi_i > 180\dg (n-2)$, it holds: $ (\varphi_0, \ldots, \varphi_{n-1}) \in \gP{n}.$
 \label{lem:P5-120}
\end{lemma}
\begin{proof}
 It is possible to choose~$\psi_i > 0$, $i=0, \ldots, n-1$ such that
 $\psi_i < \varphi_i$, $\sum_{i=0}^{n-1} \psi_i = 180\dg (n-2)$.
 Again, we set
 $\beta_i = \gamma_{i+1} = \frac{\psi_i}{2}$ in~Problem~\optprobref.
 All these angles are less than~$60\dg$, and all constraints in~\optprobref are satisfied.
 Then, the corresponding drawing in Fig.~\ref{fig:opt-problem} provides a solution.
\end{proof} 
\begin{figure}[tb]
 \hfill
 \subfloat[]{
  \includegraphics[page=4, scale=1.0]{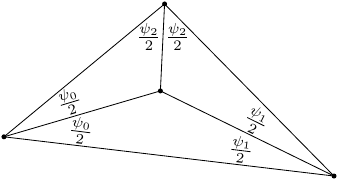}
  \label{fig:lemP3}
 }
  \hfill
  \subfloat[]{
  \includegraphics[page=5, scale=1.0]{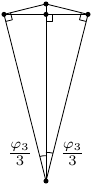}
  \label{fig:lemP4:1}
 }
 \hfill
  \subfloat[]{
  \includegraphics[page=3, scale=1.0]{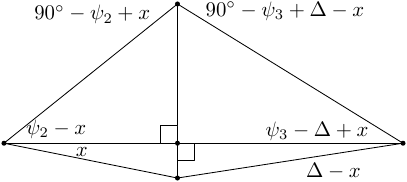}
  \label{fig:lemP4:2}
 }
  \hfill\null
 \caption{
  \protect\subref{fig:lemP3} A solution to~\gP{3}.
  \protect\subref{fig:lemP4:1},
  \protect\subref{fig:lemP4:2}: solutions of~\gP{4} for~$\varphi_0 = 180\dg$.
 }
\end{figure}
\begin{lemma}
Let~$\varphi_0$, \ldots, $\varphi_3$, $\varphi_i \geq \varphi_{i+1}>0\dg$,
 and for each~$i=0,\ldots,3$, $\varphi_i \notin (120\dg, 180\dg)$ and~$\varphi_i \notin (60\dg,90\dg]$.
 If~$\sum_{i=0}^3 \varphi_i > 360\dg$, then~$\{ \varphi_0, \ldots, \varphi_3 \} \in \gP{4}$.
 \label{lem:greedy-K14}
\end{lemma}
\begin{proof}
If~$\varphi_0 < 180\dg$,
then~$\varphi_0 \leq 120\dg$,
and the statement holds by Lemma~\ref{lem:P5-120}.
Let~$\varphi_0 = 180\dg$.
If~$\varphi_3 > 90\dg$, then a square is a solution of~$\gP{4}$.
Hence, let~$\varphi_3 \leq 60\dg$.
Then, $\varphi_1 + \varphi_2 > 120\dg$, and~$\varphi_1 > 90\dg$.
\par
If~$\varphi_2 > 90\dg$, then the construction in Fig.~\ref{fig:lemP4:1}
provides a solution.
Let now~$\varphi_2 \leq 60\dg$.
Since it is~$\sum_{i=1}^3 \varphi_i > 180\dg$,
there exist~$0\dg < \psi_i < \varphi_i$, $i=1,\ldots,3$,
and~$0\dg < \Delta < \psi_3$,
 such that~$\sum_{i=1}^3 \psi_i = 180\dg + \Delta$.
 Consider the angle assignment in Fig.~\ref{fig:lemP4:2}.
 For~$x \in (0, \Delta)$, all angles~$\beta_j$, $\gamma_j$ in Problem~\optprobref
 are in~$(0\dg, 90\dg)$, and all $\alpha_i$ are $90\dg$.
 Consider the function
 \begin{align*}
  f(x) = & \sin(\psi_2 - x) \sin(90\dg - x) \sin(\Delta - x) \sin(90\dg - \psi_3 + \Delta - x) \\
  	   - & \sin(x) \sin(90\dg - \Delta + x) \sin(\psi_3 - \Delta + x) \sin(90\dg - \psi_2 + x).
 \end{align*}
 It is~$f(0) > 0$ and~$f(\Delta) < 0$.
 Hence, for some~$x \in (0,\Delta)$ it is~$f(x)=0$.
 For this value of~$x$, the angle assignment provides a solution of~$\gP{4}$.
\end{proof}
\begin{lemma}
For each~$\varphi_4 > 0\dg$, $\{ 180\dg, 180\dg, 180\dg, 180\dg, \varphi_4\} \in \gP{5}$.
 \label{lem:4x180}
\end{lemma}
\begin{proof}
 Let~$0\dg < 16\delta < \min\{ \varphi_4, 60\dg \}$. Following angle assignment solves~\optprobref:
 \begin{align*}
  & \beta_0 = \gamma_1 = 8\delta, \;\; \gamma_0 = \beta_1 = 90\dg - 5\delta, \;\; \alpha_0 = \alpha_1 = 90\dg - 3\delta, \\
  & \beta_i = \gamma_i = 60\dg - \delta, \;\; \alpha_i = 60 + 2\delta, \;\; i = 2,\ldots,4. 
 \end{align*}
\end{proof}
\begin{figure}[b]
 \centering
 \hfill{}
 \subfloat[]{\includegraphics[page=1]{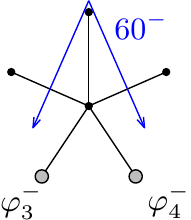} \label{fig:deg5-no2x60:1}} \hfill
 \subfloat[]{\includegraphics[page=2]{fig/full/deg5-no2x60.pdf} \label{fig:deg5-no2x60:2}}
 \hfill\null
 \caption{Proof of Lemma~\ref{lem:P5-3x180}. Both orderings are not possible if $\varphi_3 + \varphi_4 \leq 120\dg$.}
\end{figure}
\begin{lemma}
For~$\varphi_0 = \varphi_1 = \varphi_2 = 180\dg$, $\varphi_3, \varphi_4 \leq 120\dg$,
it is 
$\{ \varphi_0, \ldots, \varphi_4 \} \in \gP{5}$ if and only if $\varphi_3 + \varphi_4 > 120\dg$.
 \label{lem:P5-3x180}
\end{lemma}
\begin{proof}
 First, let~$\varphi_3 + \varphi_4 \leq 120\dg$.
 Consider the opening angles in the two embedding options in Fig.~\ref{fig:deg5-no2x60:1} and~\ref{fig:deg5-no2x60:2}. 
In the first case, angles with strict upper bounds~$60\dg$, $\varphi_3$ and~$\varphi_4$ must pairwise contain apices of each other.
In the second case, angles with strict upper bounds~$120\dg$, $\frac{\varphi_3}{2}$ and~$\frac{\varphi_4}{2}$ must form a triangle.
In both cases, the sum of the three angles is below~$180\dg$, a contradiction. \par
 Now, let~$\varphi_3 + \varphi_4 > 120\dg$, $\varphi_3, \varphi_4 \leq 120\dg$. 
 There exist~$\psi_i < \varphi_i$, $i=3,4$, and a sufficiently small~$\delta > 0$,
  such that $\psi_3 + \psi_4 - 8 \delta > 120\dg$.
 Then, following assignment satisfies~\optprobref:
 \begin{align*}
  &\beta_1 = \gamma_2 = \frac{\psi_3}{2} < 60\dg, \; \beta_3 = \gamma_4 = \frac{\psi_4}{2} < 60\dg\\
  &\beta_2 = \gamma_1 = 90\dg - \frac{\psi_3}{4}-\delta, \; \alpha_1 = \alpha_2 = 90\dg - \frac{\psi_3}{4}+\delta \\
  &\beta_4 = \gamma_3 = 90\dg - \frac{\psi_4}{4}-\delta, \; \alpha_3 = \alpha_4 = 90\dg - \frac{\psi_4}{4}+\delta \\
  &\beta_0 = \gamma_0 = 90\dg - \frac{\psi_3 + \psi_4}{4}+2\delta, \; \alpha_0 = \frac{\psi_3 + \psi_4}{2}-4\delta > 60\dg. \\
 \end{align*}
\end{proof}
\begin{figure}[t]
 \hfill
 \subfloat[]{ \includegraphics[page=1]{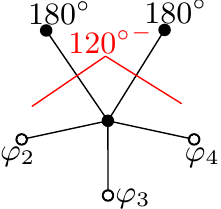} \label{fig:2x180:1} } \hfill
 \subfloat[]{ \includegraphics[page=2]{fig/full/lem-2x180.pdf} \label{fig:2x180:2} } \hfill
 \subfloat[]{ \includegraphics[page=3]{fig/full/lem-2x180.pdf} \label{fig:2x180:3} } \hfill\null
 \caption{Proof of Lemma~\ref{lem:2x180}: $\varphi_2 + \varphi_3 + \varphi_4 > 240\dg$ is necessary.
 }
 \label{fig:2x180:I}
\end{figure}
\begin{figure}[b]
 \hfill
 \subfloat{ \includegraphics[page=4]{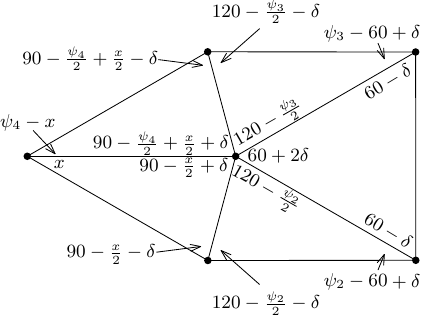}} \hfill\null
 \caption{Proof of Lemma~\ref{lem:2x180}: $\varphi_2 + \varphi_3 + \varphi_4 > 240\dg$ is sufficient.
 }
 \label{fig:2x180:4}
\end{figure}
\begin{lemma}
  Let~$\varphi_2, \varphi_3 \in (90\dg,120\dg]$ and~$\varphi_4 \in (0\dg,120\dg]$, $\varphi_2 \geq \varphi_3 \geq \varphi_4$.  Then, 
 \begin{align*}
 \{ 180\dg, 180\dg, \varphi_2, \varphi_3, \varphi_4\} \in \gP{5} \textnormal{ if and only if } 
  \varphi_2 + \varphi_3 + \varphi_4 > 240\dg. 
 \end{align*}
 \label{lem:2x180}
\end{lemma}
\begin{proof} 
 First, let~$\varphi_2 + \varphi_3 + \varphi_4 \leq 240\dg$. %
 Then, the embedding option in Fig.~\ref{fig:2x180:1} is not possible, since~${120\dg}^- + 240\dg < 360\dg$.
 Thus, the only possible embedding option is the one in Fig.~\ref{fig:2x180:2}
 for~$\{ \varphi_2, \varphi_3, \varphi_4 \} = \{ \alpha, \beta, \gamma\}$.
 Assume the solution of~$\gP{5}$ exists, and consider the corresponding construction in Fig.~\ref{fig:2x180:3}.
 It must hold
 \begin{align*}
  &\frac{\alpha}{2} + \frac{x + \beta}{2} + \frac{y+\gamma}{2} > 180\dg 
  \Rightarrow x + y > 120\dg \Rightarrow 180\dg - x - y < 60\dg,
 \end{align*}
 a contradiction to greediness. \par
 Now, let~$\varphi_2 + \varphi_3 + \varphi_4 > 240\dg$.
 Then, there exist $\psi_2, \psi_3, \psi_4$ and a sufficiently small~$\delta$ such that:
 \begin{align*}
  & \psi_i < \varphi_i, \; i=2,3,4; \\
  & \delta < \min \{ \psi_2, \psi_3, \psi_4, 30\dg \}; \\
  & \psi_j - 60\dg + \delta < 60\dg, \; j=2,3; \\
  & \psi_2 + \psi_3 + \psi_4 - 8\delta = 240\dg.
 \end{align*}
Let such $\psi_2, \psi_3, \psi_4, \delta$ be fixed.
For~$x \in (0\dg, \psi_4)$, consider the angle assignment shown in Fig.~\ref{fig:2x180:4}. It is
 $(90\dg - \frac{\psi_4}{2} + \frac{x}{2} + \delta) - (\psi_4 - x)
= 90\dg - \frac{3}{4} \psi_4 + \frac{3}{4}x + \delta > 0\dg$.
Further, it is
 $(90\dg - \frac{x}{2} + \delta) - x
= 90\dg - \frac{3}{4} x + \delta > 0\dg$.
Thus, it can be easily verified that this angle assignment satisfies all linear constraints in~\optprobref.
It remains to show the existence of~$x \in (0\dg, \psi_4)$ such that:
 \begin{align*}
 &f(x) := \sin(x) \cdot \sin(120\dg - \frac{\psi_2}{2} - \delta) \cdot \sin(\psi_3 - 60\dg + \delta) 
 \cdot \sin(90\dg - \frac{\psi_4}{2} + \frac{x}{2} - \delta) \\
 	&- \sin(\psi_4 - x) \cdot \sin(90\dg - \frac{x}{2} - \delta) \cdot \sin(\psi_2 - 60\dg + \delta)
 	\cdot \sin(120\dg - \frac{\psi_3}{2} - \delta) \stackrel{!}{=} 0.
 \end{align*}

 For~$x = 0\dg$, all angles except~$x$ are in $(0\dg,180\dg)$.
Similarly, for~$x = \psi_4$, all angles except~$\psi_4 - x$ are in $(0\dg,180\dg)$. Therefore, it is~$f(0\dg) < 0$ and~$f(\psi_4) > 0$. Thus, such~$x$ exists. 
\end{proof}
The last two cases for~$n=5$ are the only remaining ones to consider
(for~$\varphi_3 + \varphi_4 > 120\dg$, $\varphi_2 + \ldots + \varphi_4 > 240\dg$, 
$\varphi_1 + \ldots + \varphi_4 > 360\dg$).
In practice, it is possible to either strictly prove~$\{ \varphi_0, \ldots, \varphi_4 \} \notin \gP{5}$ 
or numerically construct a solution for many such sets of angles.
If we drop the last constraint in~\optprobref,
we acquire a linear program which has a constant number of variables and constraints
and can be solved in~$O(1)$. If it has no solution for any cyclic order of~$\varphi_i$, neither has~$\gP{5}$.
For example, this is the case for~$\{ {180\dg}, {105\dg}, {105\dg}, {105\dg}, 60\dg \}$.
If this linear program has a solution, we can try to solve~\optprobref
using nonlinear programming solvers. Using MATLAB we solved~$\gP{5}$ for the tree in Fig.~\ref{fig:deg45:P1}.
However, if the solver finds no solution, we obviously have no guarantee that none exists.
Lemma~\ref{lem:1x180-suff} presents a sufficient condition for the first of the two above cases.
We do not know whether it is also necessary, but interestingly, in our experiments, MATLAB found a solution %
exactly when it was satisfied.
\par
\begin{lemma}
 Consider angles~$0\dg \leq \varphi_4 \leq 60\dg$, $90\dg < \varphi_3 \leq \varphi_2 \leq \varphi_1 \leq 120\dg$,
 $\varphi_1 + \ldots + \varphi_4 > 360\dg$. 
 Let the following two conditions hold:
 \begin{enumerate}
   \item[(i)] $14 \varphi_1 + 12 \varphi_2 + 8 \varphi_3 + 15 \varphi_4 > 4500\dg$
   \item[(ii)] For $x := \min\{ \frac{1}{7} (14 \varphi_1 + 12 \varphi_2 + 8 \varphi_3 + 15 \varphi_4 - 4500\dg), \varphi_4 \}$ and
    $p_1 \in [0\dg, 90\dg]^{10}$,
    $p_1 = (\beta_0, \ldots, \beta_4, \gamma_0, \ldots, \gamma_4)$ defined as:\\
 \begin{minipage}[b]{0.50\linewidth}
\centering
     \begin{align*}
     \beta_0 &= \varphi_4-x, \\
     \beta_1 &= 90\dg - \frac{x}{2},\\
     \beta_2 &= \varphi_3 + \frac{\varphi_2}{2} + \frac{\varphi_1}{4} + \frac{\varphi_4-x}{8} - 157.5\dg, \\
     \beta_3 &= \varphi_2 + \frac{\varphi_1}{2} + \frac{\varphi_4-x}{4} - 135\dg, \\
     \beta_4 &= \varphi_1 - 90\dg + \frac{\varphi_4-x}{2}, 
    \end{align*}
\end{minipage}
\begin{minipage}[b]{0.50\linewidth}
\centering
     \begin{align*}
     \gamma_0 &= 90\dg - \frac{\varphi_4 - x}{2}, \\
     \gamma_1 &= x, \\
     \gamma_2 &= 168.75\dg - \frac{\varphi_3}{2}- \frac{\varphi_2}{4} - \frac{\varphi_1}{8} - \frac{\varphi_4 - x}{16}, \\
     \gamma_3 &= 157.5\dg - \frac{\varphi_2}{2} - \frac{\varphi_1}{4} - \frac{\varphi_4-x}{8}, \\
	 \gamma_4 &= 135\dg - \frac{\varphi_1}{2} - \frac{\varphi_4-x}{4}, 
    \end{align*}
\end{minipage}
it holds: $ \omega(p_1) < 0 $.
\end{enumerate}
Then, $\{ 180\dg, \varphi_1, \ldots, \varphi_4\} \in \gP{5}$. 
 \label{lem:1x180-suff}
\end{lemma}
\begin{figure}[tb]
 \centering
 \includegraphics[trim = 0ex 0ex 0ex 0ex, clip=false]{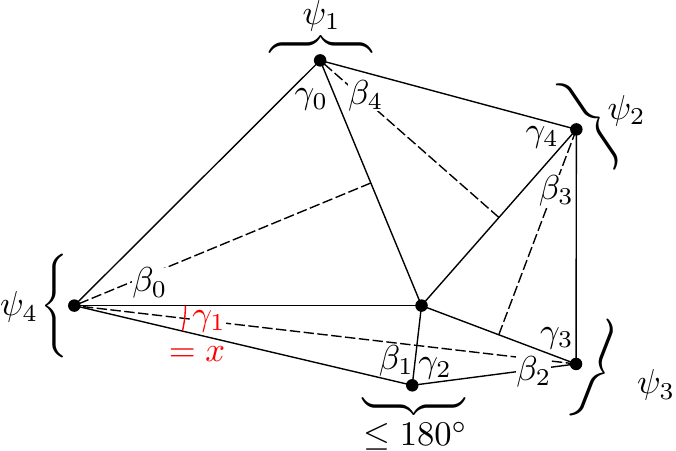}
 \caption{Proof of Lemma~\ref{lem:1x180-suff}.}
 \label{fig:1x180-suff}
\end{figure}
\begin{proof}
Assume both conditions hold.
See the construction in Fig.~\ref{fig:1x180-suff}.
The angles in~$p_1$ are chosen such that all five triangles are equilateral:
it is $\beta_1 = \alpha_1$ and~$\gamma_i = \alpha_i$ for~$i = 0,2,3,4$.
Furthermore, $\beta_0, \gamma_1, \beta_2, \beta_3, \beta_4 \leq 60\dg$.
Consider the permutation~$\tau = (4, 0, 3, 2, 1)$.
By verifying the conditions in Observation~\ref{lem:w-zero-on-simplex},
we see that~$p_1 \in S_\tau =: S$, in particular, $p_1 \in \partial S$.
\par
We now define another point~$p_2 \in [0\dg, 90\dg]^{10}$.
Due to condition~$(i)$, there must exist $90\dg < \psi_i < \varphi_i$ for $i=1,\ldots,3$, $0\dg < \psi_4 < \varphi_4$,
$0\dg < \eps < \psi_4$ (for proper~$\psi_i$, $\eps$ can be chosen arbitrarily small), such that:
$$ 14\psi_1 + 12\psi_2 + 8\psi_3 + 15\psi_4 - 80\eps = 4500\dg. $$
\par
Consider the point $p_2 = (\overline{\beta_0},\ldots,\overline{\beta_4},\overline{\gamma_0},\ldots,\overline{\gamma_4})$, such that:\\
 \begin{minipage}[b]{0.50\linewidth}
\centering
     \begin{align*}
     \overline{\beta}_0 &= \psi_4, \\
     \overline{\beta}_1 &= 90\dg - \eps,\\
     \overline{\beta}_2 &= \psi_3 + \frac{\psi_2}{2} + \frac{\psi_1}{4} + \frac{\psi_4}{8} - 157.5\dg, \\
     \overline{\beta}_3 &= \psi_2 + \frac{\psi_1}{2} + \frac{\psi_4}{4} - 135\dg, \\
     \overline{\beta}_4 &= \psi_1 - 90\dg + \frac{\psi_4}{2}, 
    \end{align*}
\end{minipage}
\begin{minipage}[b]{0.50\linewidth}
\centering
     \begin{align*}
     \overline{\gamma}_0 &= 90\dg - \frac{\psi_4}{2} - \eps, \\
     \overline{\gamma}_1 &= 0\dg, \\
     \overline{\gamma}_2 &= 168.75\dg - \frac{\psi_3}{2}- \frac{\psi_2}{4} - \frac{\psi_1}{8} - \frac{\psi_4 - x}{16} - \eps, \\
     \overline{\gamma}_3 &= 157.5\dg - \frac{\psi_2}{2} - \frac{\psi_1}{4} - \frac{\psi_4-x}{8} - \eps, \\
	 \overline{\gamma}_4 &= 135\dg - \frac{\psi_1}{2} - \frac{\psi_4-x}{4} - \eps. 
    \end{align*}
\end{minipage}
The condition~$\sum_{i=1}^5 (\overline{\beta}_i + \overline{\gamma}_i) = 540\dg$ holds, since
\begin{align*}
16 \cdot (90\dg + 168.75\dg) + (14\psi_1 + 12\psi_2 + 8\psi_3 + 15\psi_4 - 16 \cdot 5\eps) = 16\cdot 540\dg
\end{align*}
due to the choice of~$\psi_i$ and~$\eps$. The rest of the conditions for $p_2 \in S$ can be easily verified.
Apart from~$\overline{\gamma}_1 \geq 0\dg$, all inequalities are strict.
Since~$\gamma_1 > 0\dg$, for each~$\lambda \in (0,1)$,
the point $\lambda p_1 + (1-\lambda) p_2$ lies in the interior of~$S$.
Since~$\omega(p_1) < 0$ and~$\omega(p_2) > 0$ (due to $\overline{\gamma}_1 = 0\dg$),
by the mean value theorem, 
$\omega(\lambda p_1 + (1-\lambda) p_2) = 0$ for some~$\lambda \in (0,1)$.
\end{proof}
\section{Characterizing greedy-drawable binary trees}
\label{sec:binary}
In this section, we shall characterize greedy-drawable binary trees by forbidden subgraphs.
\par
Let us consider the following subtree~$Q_k$ with root~$b_0$.
It consists of nodes $b_0$, $b_1$, $c_1$, $b_2$, $c_2$, \ldots, $b_{k+1}$, $c_{k+1}$, $b_{k+2}$.
Node~$b_i$ is connected to~$b_{i-1}$, $b_{i+1}$ and~$c_i$, 
and nodes~$c_1, \ldots, c_k$ and~$b_{k+1}$ are leaves (e.g., $Q_0 = K_{1,3}$).
Hence, $(b_1, b_2, c_1)$ is an irreducible triple with parent~$b_0$,
$(b_2, b_3, c_2)$ is an irreducible triple with parent~$b_1$ etc.
Figure~\ref{fig:Qk:1} shows a subdivision of such a subtree~$Q_1$.
\begin{lemma}
 For a subdivision of~$Q_k$, an open angle~$\psi_k \geq 90^\circ+30^\circ/(2^k)$ is not possible.
 For each~$\varepsilon' > 0$, $Q_k$ can be drawn with open angle~$90^\circ+30^\circ/(2^k)-\varepsilon'$.
  \label{lem:stackedforks-maxangle} 
\end{lemma}
\begin{figure}[tb]
 \hfill
  \subfloat[]{ \includegraphics[page=3, scale=1.0]{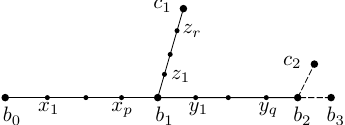} \label{fig:Qk:1} } \hfill
  \subfloat[]{ \includegraphics[page=4, scale=1.0]{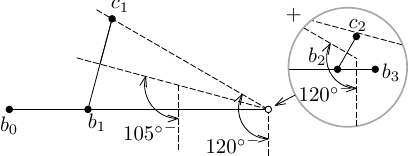} \label{fig:Qk:2} } \hfill\null
 \caption{A subdivision of~$Q_1$ and a greedy drawing with optimal opening angle.}
 \label{fig:Qk}
\end{figure}
\begin{proof}
 It is~$|\angle Q_0| < {120\dg}$. Each subdivision of~$Q_0$ contains it as a subgraph.
 Furthermore, angle~${120\dg} - \eps$ is possible for each~$\eps > 0$ 
 (draw each of the three simple paths collinear and make all edges infinitesimally small, except for the three segments 
 adjacent to the node of degree~3). \par
 Let~$T_k$ be a subdivision of~$Q_k$, $k \geq 1$.
 Let~$(b_0 = x_0, x_1, \ldots, x_p, b_1)$ be a subdivision of~$b_0 b_1$,
 $(b_1, y_1, \ldots, y_q, b_2)$ be a subdivision of~$b_1 b_2$
 and~$(b_1, z_1, \ldots, z_r, c_1)$ a subdivision of~$c_1 c_2$; see Fig.~\ref{fig:Qk:1}.
 Let~$T_{k-1} = ({T_k})_{b_1 y_1}^{y_1} + b_1 y_1$ with root $b_1$.
 Then,~$T_{k-1}$ is a subdivision of~$Q_{k-1}$,
 and by induction, $|\angle T_{k-1}| = (90\dg - \frac{30\dg}{2^{k-1}})^-$.
 Applying Lemma~\ref{lem:combine-case1} to~$(T_k)_{x_p b_1}^{b_1} + {x_p b_1}$
 and then repetitively applying Lemma~\ref{lem:combine-case0} to~$(T_k)_{x_i x_{i+1}}^{x_{i+1}} + {x_i x_{i+1}}$
 provides
 \begin{align*}
  & |\angle T_k| = ( \frac{1}{2} (90\dg - \frac{30\dg}{2^{k-1}}) + 45\dg )^-
   = (90\dg - \frac{30\dg}{2^{k}})^-.
 \end{align*}
\end{proof}
 For $\varepsilon >0$, angle $90^\circ + 30^\circ/(2^k) - \varepsilon$ is achieved if~$b_1$, $b_2$, \ldots, $b_{k+2}$ lie on a single line,
 $\angle b_{k+2} b_{k+1} c_{k+1}$ is slightly bigger than~$60^\circ$,
 $\angle b_{k+1} b_{k} c_{k}$ is slightly bigger than~$75^\circ$ etc; see Fig.~\ref{fig:Qk:2}.
 \par
It follows that subtrees of type~$Q_k$ and subdivisions thereof
 can always be drawn with an opening angle~$90 + \varepsilon_k$, $\varepsilon_k > 0$, for any fixed~$k$.
We show that if there are at most four such independent components in~$T$,
their four open angles can always be arranged appropriately.
\begin{figure}[b]
 \hfill
  \subfloat{
  \includegraphics[page=3]{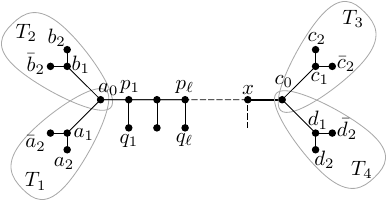}
  \label{fig:bintrees-4triples:crab}
 }\hfill\null
 \caption{Sketch of the proof of Lemma~\ref{lem:bintrees:4triples}.}
\end{figure}
\begin{lemma}
 If a binary tree~$T$ contains at most four independent irreducible separating triples ($n_3 \leq 4$), it has a greedy drawing.
 \label{lem:bintrees:4triples}
\end{lemma}
\begin{proof}
 Without loss of generality, let~$n_3 = 4$.
 Then,~$T$ contains a subdivision of a subtree depicted in Fig.~\ref{fig:bintrees-4triples:crab}. 
 (Coincidentally, for~$\ell=0$, this is exactly the ``crab'' from~\cite{acglp-sag-12}. Hence, $T$ has no self-approaching drawing.)
 Let~$T_1 = T_{a_0 a_1}^{a_1} + {a_0 a_1}$, $T_2 = T_{a_0 b_1}^{b_1} + {a_0 b_1}$ be subtrees of~$T$ rooted at~$a_0$
 and~$T_3 = T_{c_0 c_1}^{c_1} + {c_0 c_1}$, $T_4 = T_{c_0 d_1}^{d_1} + {c_0 d_1}$ be subtrees of~$T$ rooted at~$c_0$.
 Then,~$T_1$, \ldots, $T_4$ must be subdivisions of caterpillars of type~$Q_k$ (otherwise it would be~$n_3 \geq 5$).
 \par
 Now, we start combining the subtrees.
 Let~$T' = T_{c_0 x}^x + { x c_0 }$ and~$T'' = T_{x c_0}^{c_0} + { x c_0 }$.
 By applying Lemmas~\ref{lem:combine-case0} and~\ref{lem:combine-case1}
to $T_1$ and~$T_2$ as well as to $T_3$ and~$T_4$,
 it follows that both~$T'$ and~$T''$ can be drawn greedily with an opening angle~$\eps$
 for sufficiently small~$\eps > 0$.
 We merge the two drawings at the edge~$x c_0$ and gain a greedy drawing of~$T$.
\end{proof}
\begin{figure}[tb]
 \hfill
 \subfloat[]{\includegraphics[page=5, scale=0.8]{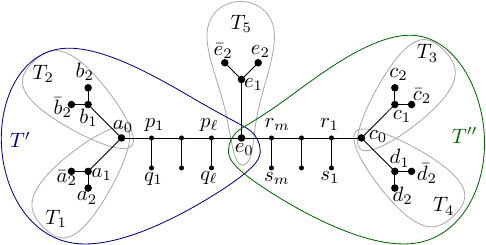}
  \label{fig:bintrees-5triples:crab}
 }
  \hspace{1cm}
  \subfloat[]{\includegraphics[page=6, scale=0.8]{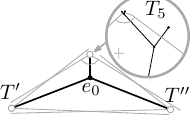}
  \label{fig:bintrees-5triples:angles}
 }\hfill\null
 \caption{Constructing a greedy drawing for~$n_5 = 5$.}
\end{figure}
We now consider the remaining case~$n_3=5$.
In this case,~$T$ must contain a \emph{five-crab} subgraph shown in Fig.~\ref{fig:bintrees-5triples:crab}
or a subdivision thereof.
We consider the corresponding independent subtrees~$T_1$, \ldots, $T_5$ of~$T$.
Again, these subtrees must be caterpillars of type~$Q_k$, otherwise,~$n_3 \geq 6$.
Each~$T_i$ can be drawn with an opening angle~$|\angle T_i| = \varphi_i^- \in (90\dg, 120\dg)$.
Let~$\sigma = \sum_{i=1}^5 \varphi_i$.
If~$\sigma \leq 540\dg$, 
no greedy drawing exists by Lemma~\ref{lem:indep-angles-sum}.
\par
We now show that a greedy drawing always exists for~$\sigma > 540\dg$.
Similar to the proof of Lemma~\ref{lem:bintrees:4triples},
we combine the subtrees~$T_1$, $T_2$ and edges~$p_1 q_1$, \ldots, $p_\ell q_\ell$ to the subtree~$T' = T_{e_0 p_\ell}^{p_\ell} + {e_0 p_\ell}$
as well as the subtrees~$T_3$, $T_4$ and edges~$r_1 s_1$, \ldots, $r_m s_m$ to the subtree~$T'' = T_{e_0 r_m}^{r_m} + {e_0 r_m}$ (both with root~$e_0$).
By applying Lemmas~\ref{lem:combine-case0}, \ref{lem:combine-case1} and~\ref{lem:combine-case3},
we have~$|\angle T'| = {\varphi'}^-$ for~$\varphi' = \varphi_1 + \varphi_2 - 180\dg$
and~$|\angle T''| = {\varphi''}^-$ for~$\varphi'' = \varphi_3 + \varphi_4 - 180\dg$.
Since~$\varphi' + \varphi'' + \varphi_5 > 180\dg$, it is 
$\{ \varphi', \varphi'', \varphi_5 \} \in \gP{3}$ by Lemma~\ref{lem:greedy-K13}.
\par
We now list all the possibilities for~$\sigma > 540\dg$.
For the rooted subtrees~$T_i$, $i=1,\ldots,5$ we say that~$T_i$ has \emph{order~$k$} if~$T_i$ is equivalent 
to a subdivision of~$Q_k$.
Assume at least four of the five subtrees have order~1 or greater, then, 
$\sigma \leq 120\dg + 4 \cdot 105\dg = 540\dg$,
so~$T$ cannot be drawn greedily.
Thus, at least two subtrees~$T_i$ have order~0.
If there are three, four or five such subtrees,
it is~$\sigma > 3 \cdot 120\dg + 2 \cdot 90\dg = 540\dg$.
If there are only two, then at least two of the three remaining subtrees have order~1, 
since~$2 \cdot 120\dg + 105\dg + 2 \cdot 97.5\dg = 540\dg$.
In this case, $\sigma > 2 \cdot 120\dg + 2 \cdot 105\dg + 90\dg = 540\dg$.
In both cases, a greedy drawing exists by Lemma~\ref{lem:P5-120}. \par
We can now give a complete characterization of greedy-drawable trees with maximum degree~3.
\begin{proposition}
 A tree~$T$ with maximum degree~3 has a greedy drawing in~$\mathbb R^2$ if and only if one of the following holds:
 \begin{enumerate}
   \item $n_3 \leq 4$, or
   \item $T$ contains a subdivision of a five-crab in Fig.~\ref{fig:bintrees-5triples:crab}, such that 
   the rooted subtrees~$T_1$, \ldots, $T_5$ as defined above are subdivisions of~$Q_k$ of orders
   either $\{ 0,0,0,x_1,x_2\}$ or $\{ 0,0,1, 1, x_1 \}$ for some~$x_1, x_2 \in \mathbb N_0$.
 \end{enumerate}
\end{proposition}
Alternatively, we can express it using forbidden subgraphs: 
1) a five-crab with four subtrees~$Q_1$ or 2) a five-crab with two subtrees~$Q_2$ and one~$Q_1$ (or subdivisions thereof).
\section{Recognition algorithm}
\label{sec:algorithm}
\textbf{Maximum degree 4.} In this section we formulate Algorithm~\ref{alg:decide-deg4}, which  decides for a tree~$T$ with maximum degree~4 whether~$T$ has a greedy drawing.
First, we describe a procedure to determine the tight upper bound for the opening angle of a given rooted subtree.
Let $\neighb(v)$ denote the neighbors of~$v \in V$ in~$T$.
After processing a node~$v$, we set a flag~$\proc(v) = \true$.
Let $\procn(v) = \{ u \mid uv \in E, \proc(u) = \true\}$,
and $\optangle$ the new tight upper bound 
calculated according to Table~\ref{tab:combine}.
\par
\setlength{\algomargin}{0.3cm}    
\begin{figure}
\centering
\begin{minipage}[t]{.46\linewidth}
\begin{procedure}[H]
  	\small
 	\DontPrintSemicolon
	\SetKwInOut{Input}{Input}
	\SetKwInOut{Output}{output}
	\Input{tree~$T=(V,E)$, root~$r \in V$, $d_r=1$}
	\KwResult{tight upper bound on $|\angle T|$, 0~if no open angle possible. }
	 $\proc(r) \leftarrow \false$\;
	\For{$v \in V \setminus \{ r \}$}{
	 \lIf{$d_v = 5$}{
	 	$\Return~0$\;
	 }
	 \uElseIf{$d_v = 1$}{
	 	$\proc(v) \leftarrow \true$\;
	 	$\angle(v) \leftarrow 180$
	 }
	 \lElse{
	 	$\proc(v) \leftarrow \false$
	 }
	}
	\While{$\exists v \in V: \neg\proc(v)$ $\And |\procn(v)| = d_v-1$}{
		\uIf{$\forall u \in \procn(v): \angle(u)=180$}
			{$\angle(v) \leftarrow 180 -(d_v-2)\cdot 60$\;}
		 \lElseIf{case I,\ldots,V applicable}{
		 	$\angle(v) \leftarrow \optangle(\procn(v))$\;
		 }
		 \lElse{
			$\Return~0$\;
		 }
		$\proc(v) \leftarrow \true$\;
	}
	$\Return~\angle(v)$ for $\{ v \} = \neighb(r)$
\vspace{5.3mm}
 \caption{getOpenAngle($T$,$r$)}
 \label{proc:getAngle}
\end{procedure}
\end{minipage}
\begin{minipage}[t]{.52\linewidth}
\begin{algorithm}[H]
    \footnotesize
	\DontPrintSemicolon
	\SetKwInOut{Input}{Input}
	\SetKwInOut{Output}{output}
	\Input{tree~$T=(V,E)$, $\max \deg~4$}
	\KwResult{whether~$T$ has a greedy drawing}
	\For{$v \in V$}{
	 \uIf{$d_v = 1$}{
	 	$\proc(v) \leftarrow \true$;
	 	$\angle(v) \leftarrow 180$\;
	 }
	 \lElse{
	 	$\proc(v) \leftarrow \false$
	 } %
	}
	\While{$\exists v \in V: \neg\proc(v)$ $\And |\procn(v)| \geq d_v-1$}{ %
		\uIf{$|\procn(v)| = d_v$}{\Return~$\sum_{u, uv \in E} \angle(u) > (d_v-2)180$}
		\uElseIf{$\forall u \in \procn(v): \angle(u)=180$}
			{$\angle(v) \leftarrow 180 -(d_v-2)\cdot 60$\;}
		 \uElseIf{case I,\ldots,V applicable}{
		 	$\angle(v) \leftarrow \optangle(\procn(v))$\;
		 }
		 \uElse{
		 	$w \leftarrow \neighb(v) - \procn(v)$\;
		 	$\angle(w) \leftarrow \ref{proc:getAngle}(T_{vw}^w+ vw,v)$\;
		 	$\Return~\angle(w) > 0$ $\And \sum _{u,uv \in E}\angle(u) > (d_v-2)180$\;
		 }
		 $\proc(v) \leftarrow \true$\;
	}
	\caption{hasGreedyDrawing($T$)}
	\label{alg:decide-deg4}
\end{algorithm}
\end{minipage}
\end{figure}
\begin{lemma}\label{lem:proc:getAngle}
 Procedure~\ref{proc:getAngle} is correct and requires time $O(|V|)$.
\end{lemma}
\begin{proof}
 The algorithm processes tree nodes bottom-up. 
 For~$v \in V$, let $\pi_v$ be the parent of~$v$, $\deg(v) = d_v$, $T_v = T_{\pi_v v}^v + \pi_v v$ with root~$\pi_v$.
 For a subtree with one or two nodes, define its opening angle as~$180\dg$.
 We prove the following invariant for the~\emph{while} loop: 
 For each~$v \in V$ with~$\proc(v) = \true$, $\angle(v) > 0$ stores a tight upper bound for the opening angle in a greedy drawing of~$T_v$. \par
 The invariant holds for all leaves of~$T$ after the initialization.
 The first \emph{if}-statement inside the \emph{while} body ensures that if all nodes in~$T_v$ except~$v$
 have degree~1 or~2, then
 $\angle(v) = 180$ if~$d_v = 1,2$ in~$T$,
 $\angle(v) = 120$ if~$d_v = 3$ and
 $\angle(v) = 60$ if~$d_v = 4$. %
 Now consider the first \emph{else} clause inside the \emph{while} loop.
 Assume $\proc(v)=\false$, $|\procn(v)|$ $= d_v-1$ and the invariant holds for all subtrees~$T_u$, $u \in \procn(v)$.
 If one of the cases~I--V can be applied to~$v$ and subtrees~$T_u$, then,
  after the current loop, $\angle(v) > 0$ stores the tight upper bound for the opening angle in a greedy drawing of~$T_v$;
 see Table~\ref{tab:combine}.
 Otherwise, we have case~VI or VII, and $T_v$ cannot be drawn with an open angle.
 Each node~$v$ is processed in~$O(d_v)$, and if for $u \in \neighb(v) - \procn(v)$,
 it holds~$|\procn(u)| \geq d_u -1$ after processing~$v$, we put~$u$ in a queue.
 Hence the running time is $O(|V|)$.
\end{proof}
\begin{proposition}
 Algorithm~\ref{alg:decide-deg4} is correct and requires time $O(|V|)$.
  \label{prop:alg-deg4}
\end{proposition}
\begin{proof}
 The algorithm is similar to Procedure~\ref{proc:getAngle}, except that~$T$ now doesn't have a distinguished root.
 We proceed from the leaves of~$T$ inwards.
 For a node~$v$ with~$|\procn(v)| = d_v-1$, let~$\{r_v\} = \neighb(v) - \procn(v)$.
 Similar to Procedure~\ref{proc:getAngle}, after~$\proc(v)$ is set~$\true$,
 $\angle(v) > 0$ stores the tight upper bound for the opening angle of subtree~$T_{r_v v}^v + r_v v$ (this is proved as in  Lemma~\ref{lem:proc:getAngle}). \par
 Let us now consider the two \emph{return} statements.
 In the first one, we have a node~$v$, and for all its neighbors~$u_i$, $i = 0, \ldots, d_v-1$, 
 $\proc(u_i) = \true$ and~$\varphi_i = \angle(u_i) >0$
 by the invariant.
 Angle~$\varphi_i$ is the tight upper bound on the opening angle for the subtree~$T_{v u_i}^{u_i} + u_i v$.
 Hence, if~$\sigma = \sum_{i=0}^{d_v-1} \varphi_i \leq (d_v - 2)180\dg$,
 by Lemma~\ref{lem:indep-angles-sum}, no greedy drawing of~$T$ exists.
 Now let~$\sigma > (d_v - 2)180\dg$.
 If~$d_v=2$, the two opening angles can be arranged in a suitable way.
 If~$d_v=3,4$, then~$\{ \varphi_0, \ldots \varphi_{d_v-1}\} \in \mathcal P^{d_v}$;
 see Table~\ref{tab:P}.
 By Theorem~\ref{thm:p5equiv}, a greedy drawing of~$T$ exists. \par
 Finally, consider the second \emph{return} statement and the last \emph{else} clause.
 Let $\{ u_0, \ldots, u_{d-2}\} = \procn(v)$ and~$\varphi_i = \angle(u_i)$.
 Again, since none of the cases I--V is applicable, 
 the combined tree~$T_{vw}^v + vw$ with root~$w$ must have a closed angle.
 Hence, if~$\angle(w)=0$, $T_{vw}^w + wv$ must also form a closed angle, and no greedy drawing exists
 by Lemma~\ref{lem:two-closed-angles}.
 Now let~$\varphi_{d_v-1} = \angle(w) > 0$, $\sigma = \sum_{i=0}^{d_v-1} \varphi_i$.
 Similar to the previous case, a greedy drawing exists iff~$\sigma > (d_v - 2)180\dg$;
 see Table~\ref{tab:P}.
\end{proof}
\textbf{Maximum degree 5 and above.}
If~$T$ contains a node~$v$ with $\deg(v) \geq 6$,
no greedy drawing exists.
Also, a greedy-drawable tree can have at most one node of degree~5 by Lemma~\ref{lem:indep-tuples-num}, 
otherwise, there are two independent 5-tuples.
\par
For unique~$r \in V$, $\deg(r)=5$,
consider the five rooted subtrees~$T_0, \dots, T_4$
attached to it and the tight upper bounds~$\varphi_i$ on~$|\angle T_i|$. %
If~$\sigma = \sum_{i=0}^4 \varphi_i \leq 540\dg$, $T$ cannot be drawn greedily.
The converse, however, doesn't hold.
By Theorem~\ref{thm:p5equiv}, a greedy drawing exists if and only 
if~$\{ \varphi_0, \ldots, \varphi_4\} \in \gP{5}$. %
To decide whether~$\{ \varphi_0, \ldots, \varphi_4\} \in \gP{5}$,
we apply the conditions from Table~\ref{tab:P}.
For the remaining case~$\varphi_0 = 180\dg$, $\varphi_1, \ldots, \varphi_4 \leq 120\dg$,
if the sufficient condition of Lemma~\ref{lem:1x180-suff} does not apply,
the linear relaxation of Problem~\optprobref has a solution,
but the non-linear solver finds none, we report \emph{uncertain};
see Algorithm~\ref{alg:decide-deg5}.
An uncertain example is presented in Fig.~\ref{fig:deg45:tree1}.
\begin{figure}[tb]
\hfill
 \subfloat[]{ \includegraphics[page=3]{fig/full/deg45.pdf} \label{fig:deg45:tree2}} \hfill
 \subfloat[]{ \includegraphics[page=1]{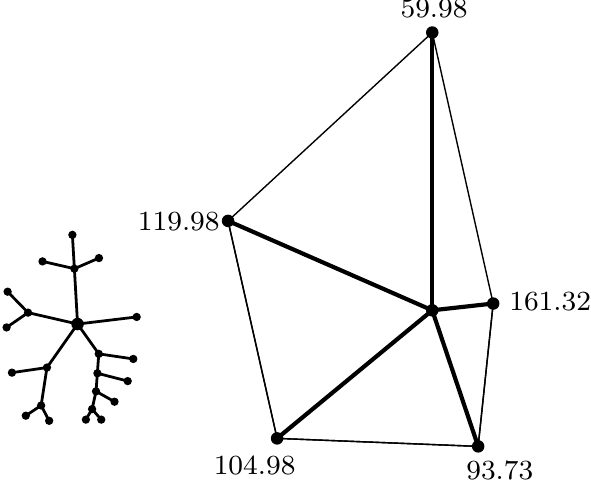} \label{fig:deg45:P1}} \hfill
 \subfloat[]{ \includegraphics[page=2]{fig/full/deg45.pdf} \label{fig:deg45:tree1}}\hfill\null
  \caption{Some trees with a node of degree~5.
   Tree \protect\ref{fig:deg45:tree2} has no greedy drawing, since $\{ 180\dg, 105\dg, 105\dg, 105\dg, 60\dg\} \notin \gP{5}$.
   Tree \protect\ref{fig:deg45:P1} has one, since $\{ 180\dg, 120\dg, 105\dg, 93.75\dg, 60\dg\} \in \gP{5}$,
   see the solution found by a non-linear solver. 
   It is not clear whether the tree in
   \protect\ref{fig:deg45:tree1} has a greedy drawing.
  By Theorem~\ref{thm:p5equiv}, proving existence is equivalent to deciding whether
  $\{ 180\dg, 120\dg, 120\dg, 120\dg, 30\dg\} \in \gP{5}$.}
\end{figure}
\begin{algorithm}[tb]
	\DontPrintSemicolon
	\SetKwInOut{Input}{Input}
	\SetKwInOut{Output}{output}
	\Input{tree~$T=(V,E)$ with maximum degree~5, $r \in V, \deg(r) = 5$.}
	\KwResult{whether~$T$ has a greedy drawing}
	\uIf{$\exists u \in V \setminus\{ r\}, \deg(u) = 5$}{$\Return~\false$\;}
	$(u_0, \ldots, u_4) \leftarrow \neighb(r)$\;
	\For{$i = 0, \ldots, 4$}{
		$\alpha_i \leftarrow \ref{proc:getAngle}(T_{r u_i}^{u_i} + r u_i, r)$ \;
		\lIf{$\alpha_i = 0$}{$\Return~\false$\;}
	}
	\uIf{$\sum_{i=0}^4 \alpha_i \leq 540$}{\Return~\false}
	$(\varphi_0, \ldots, \varphi_4) \leftarrow \textnormal{sort\_desc} (\alpha_0, \ldots, \alpha_4) $\;
	\uIf{$\varphi_0 \leq 120$}{\Return~\true}
	\uIf{$\varphi_3 = 180$}{\Return~\true}
	\uIf{$\varphi_2 = 180$}{\Return $\varphi_3 + \varphi_4 > 120$}
	\uIf{$\varphi_1 = 180$}{\Return $\varphi_2 + \varphi_3 + \varphi_4 > 240$}
        \uIf{$\varphi_4 \leq 60 \And \textnormal{condition in Lemma~\ref{lem:1x180-suff} holds }$}{$\Return~\true$}
	\uIf{$\textnormal{LP\_has\_no\_solution}$}{$\Return~\false $}
	\uIf{$\textnormal{solved $\gP{5}$ numerically}$}{$\Return~\true $}
	\tcc*[h]{cases for which we have no guarantee for~$ \notin \mathcal P^5$ yet}\;
	\Return uncertain\;
	\caption{Deciding whether a tree with maximal degree~5 has a greedy drawing.}
	\label{alg:decide-deg5}
\end{algorithm}
\section{Conclusion}
\label{sec:conclusion}
In this paper, we gave the first complete characterization of all trees that admit a greedy embedding in~$\mathbb R^2$ with the Euclidean distance metric, thereby solving the corresponding open problem stated by Angelini et al.~\cite{abf-sgdae-10}. This is a further step in characterizing the graphs that have Euclidean greedy embeddings.
One direction of future work is to develop heuristics to actually draw greedy trees with non-zero edge lengths. Some simple strategies can be derived from the proofs presented in this paper. However, optimizing the resolution of such drawings appears to be a challenging task.

To fill the gaps in the characterization of graphs with an Euclidean greedy embedding in~$\mathbb R^2$, it would be interesting to consider other graph classes, e.g., non-3-connected planar graphs with cycles.
Another challenging question is to describe graphs with planar greedy drawings.
For example, the still-open \emph{strong Papadimitriou-Ratajczak conjecture}~\cite{Papadimitriou2005} states that every 3-connected planar graph has a planar greedy drawing with convex faces.

\subsubsection*{Acknowledgments.} M.N. received financial support by the `Concept for the Future’ of KIT within the framework of the German Excellence
Initiative“. R.P. is supported by the German Research Foundation (DFG) within the Research Training Group GRK 1194 ``Self-Organizing Sensor-Actuator Networks''.

{
  \bibliography{GreedyTrees}
  \bibliographystyle{abbrv}
}  
\end{document}